\documentclass[11pt]{article}
\usepackage[utf8]{inputenc}
\usepackage[T1]{fontenc}
\usepackage[top=1in, bottom=1in, left=1in, right=1in, letterpaper]{geometry}

\usepackage{amsmath, amssymb, amsfonts, bm, cite}
\usepackage{amsthm}
\usepackage{mathtools}
\usepackage{thmtools}
\usepackage{thm-restate}
\usepackage{authblk}

\usepackage{geometry}
\usepackage{graphicx}
\usepackage{subcaption}
\usepackage{tikz}
\usetikzlibrary{shapes.geometric}
\usetikzlibrary{arrows.meta}
\usetikzlibrary{positioning}
\usepackage{enumitem}
\usepackage{comment}
\usepackage{xspace}
\usepackage{ifthen}
\usepackage{boxedminipage}
\usepackage{empheq}
\usepackage[colorlinks=true,allcolors=blue,bookmarksnumbered=true,hyperfootnotes=true]{hyperref}
\usepackage{cleveref}

\usepackage{algorithm}
\usepackage[noend]{algpseudocode}

\newtheorem{theorem}{Theorem}[section]
\newtheorem{lemma}[theorem]{Lemma}

\theoremstyle{remark}

\theoremstyle{definition}

\crefname{theorem}{Theorem}{Theorems}
\crefname{lemma}{Lemma}{Lemmas}
\crefname{proposition}{Proposition}{Propositions}
\crefname{corollary}{Corollary}{Corollaries}
\crefname{fact}{Fact}{Facts}
\crefname{definition}{Definition}{Definitions}
\crefname{remark}{Remark}{Remarks}

\crefname{section}{Section}{Sections}
\crefname{appendix}{Appendix}{Appendices}
\crefname{algorithm}{Algorithm}{Algorithms}

\newcommand{\RR}{{\mathbb R}}
\newcommand{\EE}{{\mathbb E}}
\newcommand{\PP}{{\mathbb P}}
\DeclareMathOperator{\expected}{\EE}
\DeclareMathOperator{\prob}{\PP}

\newcommand{\indic}{\boldsymbol{1}}
\newcommand{\eps}{\varepsilon}

\DeclareMathOperator{\opt}{opt}
\DeclareMathOperator{\open}{open}
\DeclareMathOperator{\conn}{conn}

\DeclareMathOperator{\poly}{poly}
\DeclareMathOperator{\cost}{cost}
\DeclareMathOperator{\dime}{\lambda}
\DeclareMathOperator{\DLA}{DLA}

\newcommand{\callr}{\mathcal{R}} 
\newcommand{\callu}{\mathcal{U}} 

\hypersetup{pdfinfo={
  Title={An O(loglog n)-Approximation for Submodular Facility Location},
  Author={Fateme Abbasi, Marek Adamczyk, Miguel Bosch-Calvo, Jarosław Byrka, Fabrizio Grandoni, Krzysztof Sornat, Antoine Tinguely},
  Keywords={approximation algorithms, facility location, submodular facility location, universal stochastic facility location}
}}

\title{An $O(\log\log n)$-Approximation for Submodular Facility Location}

\author[1]{Fateme Abbasi\thanks{fateme.abbasi@cs.uni.wroc.pl}}
\author[1]{Marek Adamczyk\thanks{marek.adamczyk@cs.uni.wroc.pl}}
\author[2]{Miguel {Bosch-Calvo}\thanks{miguel.boschcalvo@idsia.ch}}
\author[1]{Jaros\l{}aw Byrka\thanks{jby@cs.uni.wroc.pl}}
\author[2]{\mbox{Fabrizio Grandoni}\thanks{fabrizio.grandoni@idsia.ch}}
\author[2]{Krzysztof Sornat\thanks{krzysztof.sornat@idsia.ch}}
\author[2]{Antoine Tinguely\thanks{antoine.tinguely@idsia.ch}}
\affil[1]{University of Wroc\l{}aw, Poland}
\affil[2]{IDSIA, USI-SUPSI, Switzerland}

\date{}

\begin{document}

\maketitle
\thispagestyle{empty}

\begin{abstract}
\noindent In the Submodular Facility Location problem (SFL) we are given a collection of $n$ clients and $m$ facilities in a metric space. A feasible solution consists of an assignment of each client to some facility. For each client, one has to pay the distance to the associated facility. Furthermore, for each facility $f$ to which we assign the subset of clients $S^f$, one has to pay the opening cost $g(S^f)$, where $g(\cdot)$ is a monotone submodular function with $g(\emptyset)=0$.

SFL is APX-hard since it includes the classical (metric uncapacitated) Facility Location problem (with uniform facility costs) as a special case. Svitkina and Tardos [SODA'06] gave the current-best $O(\log n)$ approximation algorithm for SFL. The same authors pose the open problem whether SFL admits a constant approximation and provide such an approximation for a very restricted special case of the problem.

We make some progress towards the solution of the above open problem by presenting an $O(\log\log n)$ approximation. Our approach is rather flexible and can be easily extended to generalizations and variants of SFL. In more detail, we achieve the same approximation factor for the practically relevant generalizations of SFL where the opening cost of each facility $f$ is of the form $p_f+g(S^f)$ or $w_f\cdot g(S^f)$, where $p_f,w_f\geq 0$ are input values.

We also obtain an improved approximation algorithm for the related Universal Stochastic Facility Location problem. In this problem one is given a classical (metric) facility location instance and has to a priori assign each client to some facility. Then a subset of active clients is sampled from some given distribution, and one has to pay (a posteriori) only the connection and opening costs induced by the active clients. The expected opening cost of each facility $f$ can be modelled with a submodular function of the set of clients assigned to $f$.
\end{abstract}

\clearpage
\pagenumbering{arabic}

\section{Introduction}

In the {\sc Submodular Facility Location} problem (SFL), we are given a set $C$ of $n$ clients and set $F$ of $m$ facilities, with metric distances $d: (C\cup F)\times (C\cup F)\to \mathbb{R}_{\geq 0}$. Furthermore, we are given\footnote{As usual in this framework, we assume to have an oracle access to $g(\cdot)$: given $R\subseteq C$, we can obtain the value of $g(R)$ in polynomial time.} a monotone submodular (opening cost) function $g:2^{C}\to \mathbb{R}_{\geq 0}$ with $g(\emptyset)=0$. Notice that $g(\cdot)$ is non-negative. A feasible solution consists of an assignment $\varphi:C\to F$ of each client to some facility (we also say that $\varphi(c)$ \emph{serves} $c$). The opening cost of $f\in F$ in this solution is $g(\varphi^{-1}(f))$. The cost of the solution, that we wish to minimize, is the sum of the distances from each client to the corresponding facility plus the total opening cost of the facilities, in other words
$$
\cost(\varphi)=\sum_{c\in C}d(c,\varphi(c))+\sum_{f\in F}g(\varphi^{-1}(f)).
$$
SFL captures practical scenarios where the cost of opening a facility is a (non-linear, still ``tractable'') function of the set of served clients. As we will discuss, SFL is also closely related to certain stochastic optimization problems which recently attracted a lot of attention (see, e.g, \cite{AGLW17,GGLMSS13,GGLS08,GPRS11,IKMM04} and references therein). In particular, there are scenarios where one has to pay (a posteriori) the connection and opening costs related only to a random subset of \emph{activated} clients, and this naturally induces objective functions with submodular opening costs.

SFL is APX-hard since it includes the classical {\sc Facility Location} problem (with uniform facility costs) as a special case \cite{guha1999greedy}. Hence the best we can hope for, in terms of approximation algorithms, is a constant approximation. Finding such an approximation algorithm is explicitly posed as an open problem, e.g., by Svitkina and Tardos \cite{ST10}. The same authors present an $O(\log n)$ approximation for a generalization of SFL where each facility $f$ has a distinct submodular function $g_f(\cdot)$ (and this result is tight for this generalization due to a reduction from {\sc Set Cover} by Shmoys, Swamy and Levi \cite{SSL04}). Svitkina and Tardos also present a constant approximation for a rather restrictive (still practically motivated) special case of SFL where $g(\cdot)$ is induced by certain subtrees of a node-weighted tree over the clients.

\subsection{Our Results and Techniques}

We make some progress towards the resolution of the mentioned open problem by presenting an improved approximation algorithm for SFL.
\begin{theorem}\label{thr:main}
  There is a polynomial-time $O(\log\log n)$-approximation algorithm for SFL.
\end{theorem}
Our approach is surprisingly simple (modulo exploiting some non-trivial results in the literature). By standard reductions (see Section \ref{sec:preliminaries}) we can assume that $N=n+m$ is polynomial in $n$, hence it is sufficient to provide an $O(\log\log N)$ approximation. Our starting point is a natural (configuration) LP relaxation for the problem:
\begin{align}\label{eq:configuration-lp}\tag{Conf-LP}
        \min &\sum_{f \in F} \sum_{R \subseteq C} g(R) \cdot x_R^f
        + \sum_{c\in C} \sum_{f\in F} \sum_{R\ni c} d(c,f) \cdot x_R^f & \\
        \text{s.t.} &
        \sum_{f\in F}\sum_{R\ni c} x_R^f= 1 &\forall c \in C;\nonumber \\
        & \sum_{R \subseteq C} x_R^f= 1 & \forall f \in F;\nonumber
        \\
        & x_R^f \geq 0 & \forall R\subseteq C,\ \forall f\in F.\nonumber
\end{align}
In an integral solution, we interpret $x^f_R=1$ as assigning exactly the set of clients $R$ to the facility $f$. Notice that we impose $\sum_{R \subseteq C} x_R^f = 1$. This is w.l.o.g. since $g(\emptyset)=0$ (intuitively, $x^f_\emptyset=1$ means that no client is assigned to $f$).
We can solve the above LP in polynomial time (see Section \ref{sec:omitted}).
\begin{lemma}\label{lem:solveConfLP}
  In $\poly(N)$ time one can find an optimal solution to \eqref{eq:configuration-lp} with $\poly(N)$ non-zero entries.
\end{lemma}

Given an optimal solution $\dot{x}=(\dot{x}^f_R)_{f\in F,R\subseteq C}$ to \eqref{eq:configuration-lp} of cost $\cost(\dot{x})$ as in Lemma \ref{lem:solveConfLP}, we proceed with two main stages. In the first stage (discussed in Section \ref{sec:reducingConnection}) we simply sample partial assignments of clients to facilities with the distribution induced by $\dot{x}$ for $\ln\ln N$ many times. This cost at most $\ln\ln N$ times the optimal LP cost in expectation, and leads to a partial solution that covers a random subset $C_1\subseteq C$ of clients.

In the second stage (discussed in Section \ref{sec:treeCase}) we take care of the remaining uncovered clients $C_2=C\setminus C_1$. Let us consider the restriction $\ddot{x}$ of $\dot{x}$ to $C_2$. The opening cost of $\ddot{x}$ might be as large as the opening cost of $\dot{x}$. However, in expectation, the connection cost of $\ddot{x}$ is only a $1/\ln N$ fraction of the connection cost of $\dot{x}$ (as we will show).

At this point, using the probabilistic tree embedding algorithm in \cite{FakcharoenpholRT04}, we embed the original metric $d$ into a (rooted) tree metric $d^T$ over a hierarchically well-separated tree (HST) $T$ (see Section \ref{sec:preliminaries} for the details). The opening cost of $\ddot{x}$ w.r.t.\ to the new tree instance does not change, while its connection cost grows by a factor at most $O(\log N)$ in expectation. Altogether we obtain a feasible fractional solution $\ddot{x}$ over the tree instance whose expected cost is at most $O(\cost(\dot{x}))$. Hence it is sufficient to develop an $O(\log\log N)$-approximate LP-rounding algorithm for the considered tree instance.

The next step is at the heart of our approach. Using the properties of HSTs and losing a constant factor in the approximation, we can further reduce our SFL tree instance to the following {\sc Descendant-Leaf Assignment} problem (DLA): the facilities are leaves of $T$ and the clients are arbitrary nodes of $T$. Each client $c$ must be served by a facility contained in the subtree $T_c$ rooted at $c$. The opening cost of each facility is given by $g(\cdot)$, and there are no connection costs at all. The latter problem is very similar to a problem recently studied by Bosman and Olver \cite{BO20} in the framework of {\sc Submodular Joint Replenishment} and {\sc Inventory Routing} problems. In particular, we can adapt their approach to achieve the desired $O(\log\log N)$ approximation for our DLA instance.

We remark that we do not know how to get an $O(1)$ approximation for SFL on trees (even on HSTs). Though such approximation would not imply an $O(1)$ approximation for SFL with our approach (due to the first stage), finding it seems to be a natural intermediate problem to address.

\subsection{Generalizations and Variants}

Our basic approach is rather flexible, and it can be applied to generalizations and variants of SFL. We next describe some other applications of our approach, and we expect to see a few more ones in the future. For example, we can handle the case where the opening cost of the facility $f$ is $g_f(S^f)=w_f\cdot g(S^f)$, where $w_f\geq 0$ is some input value: we call this the SFL {\sc with Multiplicative Opening Costs} problem ({\sc multSFL}).

\begin{theorem}\label{thr:mainMultiplicative}
  There is a polynomial-time $O(\log\log n)$-approximation algorithm for {\sc multSFL}.
\end{theorem}

Similarly, we can address the SFL {\sc with Additive Opening Costs} problem ({\sc addSFL}), where $g_f(S^f)=p_f+g(S^f)$ for $S^f\neq \emptyset$, $g_f(\emptyset)=0$, and $p_f\geq 0$ is some input value.
\begin{theorem}\label{thr:mainAdditive}
  There is a polynomial-time $O(\log\log n)$-approximation algorithm for {\sc addSFL}.
\end{theorem}
The above generalizations are discussed in \cref{sec:extensions}.
We remark that we do not know how to obtain an $O(\log\log n)$-approximation for the {\sc Affine} SFL case, where the opening costs are submodular functions of the form $g_f(S^f)=p_f+w_f\cdot g(S^f)$. Notice that this generalizes both {\sc addSFL} and {\sc multSFL}. This is left as an interesting and practically relevant open problem.

As mentioned earlier, SFL is closely related to stochastic variants of {\sc Facility Location}. In particular, our approach also extends to the following {\sc Universal Stochastic Facility Location} problem ({\sc univFL}). Here we are given clients $C$ and facilities $F$ with metric distances $d$ like in SFL, plus an opening cost $w_f$ for each $f\in F$. Furthermore, we have an oracle access to a probability distribution $\pi: 2^C\to \mathbb{R}_{\geq 0}$ specifying the probability $\pi(A)$ that a given subset of clients $A\subseteq C$ is activated. A feasible solution is an (universal) mapping $\varphi:C\to F$. The cost of $\varphi$ w.r.t.\ clients $A\subseteq C$ is $\cost_A(\varphi)=\sum_{c\in A}d(c,\varphi(c))+\sum_{f\in F:\varphi^{-1}(f)\cap A\neq \emptyset}w_f$. In words, this is the cost of connecting clients in $A$ to the corresponding facilities, plus the cost of opening the facilities that serve at least one client in $A$.
Our goal is to minimize $\mathbb{E}_{A\sim \pi}[\cost_A(\varphi)]$. The main motivation for \emph{universal} problems of this type is to allow a very quick (possibly distributed) reaction to requests that arrive over time. Let $\opt \colon C\to F$ minimize $\expected_{A\sim \pi}[\cost_A(\opt)]$, in other words $\opt$ is an optimal (universal) mapping. We say that an algorithm for {\sc univFL} is $\alpha$-approximate\footnote{In Section \ref{sec:related} we describe alternative ways to define the approximation ratio.} if it returns a universal mapping $\varphi$ satisfying $\mathbb{E}_{A\sim \pi}[\cost_{A}(\varphi)]\leq \alpha\cdot \mathbb{E}_{A\sim \pi}[\cost_A(\opt)]$.

Notice that the objective function of {\sc univFL} can be rewritten as
$$
  \sum_{c\in C}d(c,\varphi(c))\cdot \prob_{A\sim \pi}[\{c\}\cap A\neq \emptyset]+\sum_{f\in F}w_f\cdot \prob_{A\sim \pi}[\varphi^{-1}(f)\cap A\neq \emptyset].
$$
Hence {\sc univFL} is almost identical to SFL since $g(R) = \prob_{A\sim \pi}[R \cap A\neq \emptyset]$ is a monotone submodular function of $R$ which is $0$ for $R=\emptyset$. We can therefore adapt our techniques to achieve the following result (see \cref{sec:universal}). Let $\pi_{\min}\coloneqq \min_{c\in C}\{\prob_{A\sim \pi}[c\in A]\}$ be the smallest probability of any client to be activated. W.l.o.g. we will assume $\pi_{\min}>0$.
\begin{theorem}\label{thr:universalFL}
  There is a polynomial-time $O(\log\log \frac{n}{\pi_{\min}})$-approximation algorithm for the {\sc Universal Stochastic Facility Location} problem.
\end{theorem}
For a comparison, Adamczyk, Grandoni, Leonardi and W\l{}odarczyk \cite{AGLW17} obtain an $O(\log n)$ approximation which also holds for non-metric distances. In the case of metric distances, they obtain an $O(1)$ approximation but only in the \emph{independent activation case}, i.e., when the sampled set $A$ of active clients is obtained by independently sampling each client $c$ according to some input probability $\pi'(c)$ for $k$ times.

\subsection{Lower Bounds}

We believe that it is plausible that SFL admits a constant approximation. In particular, one might consider greedy algorithms. In \cref{app:lower-bound} we consider a natural set-cover type greedy algorithm {\tt GreedySFL} for SFL. The same algorithm gives a $1.861$-approximation when applied to the classical {\sc Facility Location} problem~\cite{JMMSV03}. We show that this algorithm does not produce $O(1)$ approximate solutions for SFL, and in fact, it is not better than our algorithm from Theorem~\ref{thr:main}.

\begin{theorem}\label{thm:lower-bound}
  {\tt GreedySFL} has an approximation ratio at least $\Omega(\log\log n)$.
\end{theorem}

The lower bound is based on the construction of a family of instances, parameterized by $\dime \in \mathbb{N}$, where the underlying metric is defined on a hypercube of dimension $\dime$.
We define $\dime \cdot 2^{\dime}$ clients and $O(\dime \cdot 2^{\dime})$ facilities.
The cost of a feasible solution is $2^{\dime}$, but the cost of the greedy algorithm is $\Theta(\log(\dime) \cdot 2^{\dime})$, hence the lower bound follows.
The main challenge in our construction is to define a submodular function $g(\cdot)$ on all subsets of clients.
We present a compact way of defining $g(\cdot)$: it is the expected value of a random process that depends on carefully designed distances (induced by edge weights) in the hypercube.

\subsection{Related Work}\label{sec:related}

In the {\sc (Metric Uncapacitated) Facility Location} problem (FL) we are given a set of clients and a set of facilities in a metric space $d$, where each facility has an opening cost $o_f$.
One has to select a subset of facilities $F'\subseteq F$ and assign each client $c$ to the closest facility $F'(c)$ in $F'$ so as to minimize $
\sum_{c\in C}d(c,F'(c))+\sum_{f\in F'}o_f$.
FL is a special case of both {\sc addSFL} and {\sc multSFL} (and of SFL in the case of uniform opening costs).
FL is among the best-studied problems in the literature from the point of view of approximation algorithms (see, e.g., \cite{CG05,MYZ06,STA97}). It is known to be APX-hard~\cite{guha1999greedy} and the current best-known 1.488-approximation algorithm~\cite{li20131} is a randomized combination of the greedy JMS algorithm~\cite{JMMSV03} with an LP-rounding algorithm from~\cite{byrka2010optimal}. Lagrangian-multiplier preserving algorithms for FL are at the heart of several approximation algorithms for fundamental clustering problems, including {\sc $k$-Median} \cite{ANSW20,BPRST17,CEMN22,CGLS23,GPST23,JMMSV03,JV01,LS16} and {\sc $k$-Means} \cite{ANSW20,CEMN22,GORSV22}.

Various variants of FL were studied in the literature and for most of them (at least with metric connection costs) a constant approximation was eventually discovered. A notable example is the {\sc Capacitated Facility Location} problem in which the number of clients that can be served from a facility is restricted by a location-specific bound. A local-search-based constant approximation for the latter problem is given in \cite{zhang2005multiexchange} (see also \cite{an2017lp} for a more recent LP-based result). SFL is one of the most natural generalizations of (metric) FL where a constant approximation is still not known.

Grandoni, Gupta, Leonardi, Miettinen, Sankowski, and Singh \cite{GGLMSS13}, among other universal stochastic problems, studied {\sc univFL} in the independent activation case. However, they compare the cost of their solution with $\expected_{A\sim \pi}[\cost_A(\opt(A))]$, where $\opt(A)$ is the optimal facility location solution restricted to clients $A$ (while we compare with $\expected_{A\sim \pi}[\cost_A(\opt)]$). For this setting they obtain a $O(\log n)$ approximation, which also holds for non-metric connection costs.

Gupta, P{\'a}l, Ravi, and Sinha \cite{GPRS11} consider a 2-stage stochastic version of FL. Here in a first stage, one buys some facilities, then a subset of active clients is sampled from a given distribution. Finally, one can buy some more facilities, however at an opening cost which is increased by a multiplicative \emph{inflation factor} $\sigma$. For this setting they present a constant approximation.

Universal stochastic problems have a natural online stochastic counterpart. For example, in the {\sc Online Stochastic Facility Location} problem clients are sampled one by one, and when client $c$ is sampled one has to connect $c$ to an already open facility or open a new facility $f$ and connect $c$ to $f$. Garg, Gupta, Leonardi and Sankowski \cite{GGLS08} consider this problem in the independent activation case, i.e. when the next client to be served is sampled from a probability distribution $\pi:C\to \mathbb{R}_{\geq 0}$. For this setting, they present an $O(1)$ approximation. Meyerson \cite{Meyerson01} studied a variant of the problem where an adversary chooses the set of input clients, and then a random permutation of them is presented in input (\emph{random order model}).

\subsection{Preliminaries and Notation}\label{sec:preliminaries}

We use $\ln$ and $\log$ for the logarithm with base 2 and $\ln$ for the natural logarithm.
Define $X = C \cup F$, and $N=|X|=|C \cup F|$. Given a metric $d$ over $X$, we let $d_{\min}$ be the smallest non-zero distance and $d_{\max}$ be the largest distance (that we assume to be positive w.l.o.g). We use $g(c)$ as a shortcut for $g(\{c\})$.

We sometimes express a feasible solution to SFL in the form $S=(S^f)_{f\in F}$, where $S^f\subseteq C$ specifies the clients $\varphi^{-1}(f)$ assigned to $f$. Notice that for each $c\in C$ there is precisely one $f\in F$ with $c\in S^f$. We define a \emph{partial assignment} as $S=(S^f)_{f\in F}$, where $S^f\subseteq C$. We say that $S$ covers the clients $C'=\cup_{f\in F}S^f\subseteq C$. Notice that, for technical reasons, in a partial assignment we allow $S^f\cap S^{f'}\neq \emptyset$ for two distinct $f,f'\in F$ (i.e. we allow to simultaneously assign a client to more than one facility). The cost of a (partial) assignment $S$ of the above type is defined as
$
  \cost(S)\coloneqq \conn(S)+\open(S),
$
where $\conn(S)\coloneqq \sum_{f\in F}\sum_{c\in S^f}d(c,f)$ is the connection cost of $S$ and $\open(S)\coloneqq \sum_{f\in F}g(S^f)$ is the opening cost of $S$. Given a (possibly infeasible) fractional solution $x$ for \eqref{eq:configuration-lp}, we analogously define $\cost(x)=\conn(x)+\open(x)$, where $\conn(x)=\sum_{c\in C} \sum_{f\in F} \sum_{R\ni c} d(c,f) \cdot x_R^f$, and $\open(x)=\sum_{f \in F} \sum_{R \subseteq C} g(R) \cdot x_R^f$.

It is convenient to define the merge $S=S_1+S_2$ of two partial assignments $S_1$ and $S_2$ naturally as follows: (1) for each facility $f\in F$, we initially set $S^f\coloneqq S^f_1\cup S^f_2$; (2) while there exist two distinct facilities $f$ and $f'$ with $S^{f}\cap S^{f'}\neq \emptyset$, replace $S^{f'}$ with $S^{f'}\setminus S^f$ (intuitively this second step guarantees that each client is assigned to no more than one facility). We observe that merging two partial assignments cannot increase the total cost.

\begin{lemma}\label{lem:merge}
  For any two partial assignments $S_1$ and $S_2$, $\cost(S_1+S_2)\leq \cost(S_1)+\cost(S_2)$.
\end{lemma}
\begin{proof}
Let $S=S_1+S_2$, and $S'$ be the intermediate value of $S$ obtained by executing only step (1) of the merge operation. One has $\conn(S')=\conn(S_1)+\conn(S_2)$. Furthermore, by the submodularity (hence subadditivity) of $g(\cdot)$, $\open(S')\leq \open(S_1)+\open(S_2)$.
Clearly $\conn(S)\leq \conn(S')$, and the monotonicity of $g(\cdot)$ implies that $\open(S)\leq \open(S')$. The claim follows.
\end{proof}
We will exploit the following fairly standard reductions (proofs in Section \ref{sec:omitted}), thanks to which in the following it will be sufficient to obtain an $O(\log\log N)$ approximation for SFL. In order to distinguish between distinct instances $J$ of the problem, we use $\cost_J(\varphi)$ to denote the cost of $\varphi$ w.r.t.\ $J$ and define similarly $\open_J(\varphi)$ etc.
\begin{lemma}\label{lem:reduceNumberFacilities}
  There is a $3$-approximate reduction from SFL to the special case where $m=n$.
\end{lemma}

\begin{lemma}\label{lem:reduceDistances}
  For any constant $\eps>0$, There is a $(1+4\eps)$-approximate reduction from SFL to the special case where the metric $d$ satisfies $d_{\min}=2$ and $d_{\max}\leq \frac{2nN}{\eps}$.
\end{lemma}

One of the key tools that we use is the notion of probabilistic tree embedding, which we use to map the input metric into a metric on a \emph{hierarchically well-separated tree} (HST) while stretching the distances by a small enough factor. We recall that an HST is an edge weighted rooted tree where all the leaves are at the same distance from the root $r$. Furthermore, on every path from a leaf to $r$ the edge weights are $1,2,4,\dots$ In particular, edges at the same level have the same weight. We will use the following construction\footnote{We slightly and trivially extend their claim to consider nodes at distance $0$.} by Fakcharoenphol, Rao and Talwar~\cite{FakcharoenpholRT04}.

\begin{theorem}[FRT metric tree embedding~\cite{FakcharoenpholRT04}]\label{thm:frt-tree-embedding}
    For any finite metric space $(M,d)$ with $d_{\min}>1$, there exists a randomized polynomial-time algorithm returning an HST $T$ such that:
    \begin{enumerate}\itemsep0pt
        \item Every $a\in M$ is mapped to some leaf $v(a)$ of $T$ (with elements at distance zero being mapped to the same leaf);
        \item Let $d^T(v(a),v(b))$ be the length of the path between the leaves $v(a)$ and $v(b)$ of $T$. Then $d^T\left(v(a),v(b)\right) \geq d(a,b)$ and $\expected\left[d^T(v(a),v(b))\right] \leq 8 \log |M| \cdot d(a,b)$;
        \item $T$ has depth $O(\log d_{\max})$.
    \end{enumerate}
\end{theorem}

For a given set $C$, let $h\colon 2^C\to \RR$ be a
monotone submodular function with $h(\emptyset)=0$. The \emph{Lov\'{a}sz extension} $\hat{h}\colon [0,1]^C \rightarrow \RR$ of $h(\cdot)$ is defined as
\begin{equation}\label{eq:lovasz-extension}
  \hat h(y) \coloneqq \min\Big\{\;\sum_{R \subseteq C}h(R)\mu_R:\ \sum_{R\subseteq C}\sum_{R \ni c}\mu_R = y_c \ \forall c\in C,\ \sum_{R \subseteq C} \mu_R=1,\ \mu\geq0 \;\Big\}.
\end{equation}
The function $\hat h(\cdot)$ is convex. We remark that $\hat h(y)$ can be alternatively defined as
\begin{equation}\label{eq:lovasz-extension-bis}
    \hat h(y) \coloneqq
        \sum_{k=1}^{n-1} h\left(\{c_1, \dots, c_k\}\right)
        (y_{c_k} - y_{c_{k+1}}) + h(C) y_{c_{n}}
\end{equation}
where the components of $y$ are sorted in decreasing order, i.e. $y_{c_1} \geq y_{c_2} \geq \dots \geq y_{c_{n}}$ \cite[Section~6.3]{fujishige2005submodular}.
By the monotonicity of $h(\cdot)$, $\hat h(\cdot)$ is also non-decreasing in the sense that $\hat h(y) \geq \hat h(y')$ if $y\geq y'$.

\section{Reducing the Connection Cost}\label{sec:reducingConnection}

In this section, we show how to compute a random partial assignment $S_1=(S_1^f)_{f\in F}$ covering a random subset of clients $C_1\coloneqq \cup_{f\in F}S_1^f\subseteq C$ with the following high-level properties: the expected cost of $S_1$ is ``small enough'' and (2) each client belongs to $C_1$ with ``large enough'' probability. In the next section, we will describe a different partial assignment $S_2=(S_2^f)_{f\in F}$, again of small enough cost, covering the remaining clients $C_2\coloneqq C\setminus C_1$. By merging these two partial assignments we obtain a feasible solution for the input problem of small enough total cost.

Let $\dot{x}$ be an optimal solution to \eqref{eq:configuration-lp} with at most $\poly(N)$ non-zero entries that can be computed via Lemma \ref{lem:solveConfLP}. The basic idea behind the next lemma is fairly standard: we sample partial assignments according to the distribution induced by $\dot{x}$ for $\ln\ln N$ times, and merge them together.

\begin{lemma}\label{lem:sampling}
  In polynomial time one can compute a random partial assignment $S_1$ covering a random subset of clients $C_1$ such that: (1) $\expected\left[\cost(S_1)\right] \leq \ln\ln (N) \cdot \cost(\dot{x})$ and (2) For each $c\in C$, $\prob[c\in C_1]\geq 1-\frac{1}{\ln N}$.
\end{lemma}
\begin{proof}
For $i \in \{1,2,\dots,\ln\ln N\}$ and for every $R\subseteq C$, we define a partial assignment $S(i,R)$ by setting $S^f(i,R)=R$ independently with probability $\dot{x}_{R}^f$ and $S^f(i,R)=\emptyset$ otherwise. Let $S_1=\sum_{i=1}^{\ln\ln N}\sum_{R\subseteq C}S(i,R)$ be obtained by merging all these solutions, and let $C_1=\cup_{f\in F}S_1^f$. Observe that
$$
  \prob[c\notin C_1]=\prod_{f\in F}\prod_{R\ni c}(1-\dot{x}_R^f)^{\ln\ln N}\leq e^{-\ln\ln N\sum_{f\in F}\sum_{R\ni c}\dot{x}_S^f}\leq e^{-\ln\ln N}=\frac{1}{\ln N}.
$$
Furthermore, by \cref{lem:merge}, $\expected[\cost(S_1)]$ is upper-bounded by
$$
\sum_{i=1}^{\ln\ln N}\sum_{R\subseteq C}\expected[\cost(S(i,R))]=\ln\ln N\cdot \sum_{f\in F,R\subseteq C}\dot{x}^f_{R}\cdot \bigg(g(R)+\sum_{c\in R}d(c,f)\bigg)= \ln\ln N\cdot \cost(\dot{x}).\qedhere
$$
\end{proof}
Consider the partial assignment $S_1$ covering the random subset of clients $C_1$ as in the previous lemma. Let $C_2\coloneqq C\setminus C_2$ be the remaining (uncovered) clients.
Let also $\ddot{x}$ be $\dot{x}$ restricted to $C_2$, i.e. $\ddot x_R^f = \sum_{R' \subseteq C_1} \dot x_{R \cup R'}^f$ for $R \subseteq C_2$ and $f\in F$. The following lemma upper bounds the expected opening and connection cost of $\ddot{x}$.

\begin{lemma}\label{lem:costResidualClients}
  One has $\open(\ddot{x})\leq \open(\dot{x})$ and $\expected[\conn(\ddot{x})]\leq \frac{1}{\ln N}\conn(\dot{x})$.
\end{lemma}
\begin{proof}
    We have $\open(\ddot{x})\leq \open(\dot{x})$ by the monotonicity of $g(\cdot)$. For the connection cost, notice that the probability of a client $c$ being in $C_2$ is at most $1/\ln N$, and only in that case one has to pay the associated connection cost. Thus by linearity of expectation, the expected connection cost of $\ddot{x}$ is at most $\conn(\dot{x})/\ln N$. The claim follows.
\end{proof}
Notice that $\ddot{x}$ is a feasible fractional solution for~\eqref{eq:configuration-lp} limited to $C_2$.
In the following section, we show how to randomly round $\ddot{x}$ to a partial assignment $S_2$ which covers $C_2$ at expected cost $O(\log\log N) \cdot \cost(\ddot{x})$. It will then follow that $S_1+S_2$ is a feasible $O(\log\log N)$-approximate solution to the input SFL instance.

\section{Approximating SFL on an HST}\label{sec:treeCase}

We say that an instance $(C\cup F,d^T,g(\cdot))$ of SFL is of \emph{HST-type} if the metric $d^T$ is an HST metric over a tree $T$ of the type guaranteed by Lemma \ref{lem:sampling}. We remark that we allow multiple clients $C(v)$ and facilities $F(v)$ to be colocated at each leaf $v$ of $T$. In this section we will describe an $O(\log\log N)$-approximate LP-rounding algorithm for the considered instances w.r.t.\ \eqref{eq:configuration-lp}.
\begin{lemma}\label{lem:treeLProunding}
  Given a feasible fractional solution $x$ to \eqref{eq:configuration-lp} for an HST-type SFL instance, in polynomial time one can compute a feasible (integral) solution for the same instance with cost at most $O(\log\log N)\cdot \cost(x)$.
\end{lemma}
Theorem \ref{thr:main} directly follows.
\begin{proof}[Proof of Theorem \ref{thr:main}]
By Lemma \ref{lem:reduceNumberFacilities} it is sufficient to describe an $O(\log\log N)$-approximation. Furthermore by Lemma \ref{lem:reduceDistances}, we can assume that $d_{\min}=2$ and $d_{\max}\leq \frac{2nN}{\eps}$.

By applying the construction of Section \ref{sec:reducingConnection} we compute a random partial assignment $S_1=(S^f_1)_{f\in F}$ covering the clients $C_1=\cup_{f\in F}S^f_1$ with expected cost at most $O(\log\log N)\cdot \cost(\dot{x})$, where $\dot{x}$ is an optimal solution to \eqref{eq:configuration-lp}. Furthermore, by Lemma \ref{lem:costResidualClients}, we obtain a feasible solution $\ddot{x}$ to \eqref{eq:configuration-lp} restricted to clients $C_2\coloneqq C\setminus C_1$ which satisfies $\open(\ddot{x})\leq \open(\dot{x})$ and $\expected[\conn(\ddot{x})]\leq \frac{1}{\ln N}\conn(\dot{x})$. By applying the probabilistic tree embedding from Theorem \ref{thm:frt-tree-embedding} to the metric $(C_2\cup F,d)$, we obtain an HST-type SFL instance $(C_2\cup F,d^T,g(\cdot))$ where the tree has depth $D=O(\log d_{\max})=O(\log N)$. Observe that $\ddot{x}$ is a feasible fractional solution for \eqref{eq:configuration-lp} restricted to $C_2$ on the HST-type instance. Furthermore, let $\conn_{T}(\ddot{x})$ denote the connection cost of $\ddot{x}$ w.r.t.\ the HST-type instance, and define similarly $\open_{T}(\ddot{x})$ and $\cost_{T}(\ddot{x})$. Then one has
$$
  \expected[\cost_{T}(\ddot{x})]=\open(\ddot{x})+\expected[\conn_{T}(\ddot{x})]\leq \open(\dot{x})+O(\log N)\cdot\expected[\conn(\ddot{x})] \leq O(\cost(\dot{x})).
$$
By applying the LP-rounding algorithm from Lemma \ref{lem:treeLProunding} to $\ddot{x}$ one obtains a partial assignment $(S^f_2)_{f\in F}$ covering the clients $C_2$ of cost at most $O(\log\log N)\cost(\dot{x})$. The same solution has no larger cost in the original problem (on a non-tree metric). Altogether $S_1+S_2$ is a feasible solution to the input SFL problem of expected cost at most $O(\log\log N)\cdot \cost(\dot{x})\leq O(\log\log N)\cdot\cost(\opt)$.
\end{proof}

In the rest of this section, we prove Lemma \ref{lem:treeLProunding}. To this aim, we will first present a reduction to a different problem that we call the {\sc Descendent-Leaf Assignment} problem (DLA) (see Section \ref{sec:reductionDLA}). Then, we will present a good-enough approximation algorithm for DLA (see Section \ref{sec:approximationDLA}).

\subsection{A Reduction to DLA}\label{sec:reductionDLA}

In the {\sc Descendent-Leaf Assignment} problem (DLA) we are given a rooted tree $\tilde{T}$ with depth $D$, a set of facilities $\tilde{F}$ and a set of clients $\tilde{C}$. Each $x\in \tilde{F}\cup \tilde{C}$ is mapped into some node $v(x)$ of $\tilde{T}$, with the restriction that facilities are mapped into leaves of $\tilde{T}$.
By $\tilde{F}_c$ we denote the facilities which are assigned to nodes that are descendants of $v(c)$ in $T$ ($v(c)$ included if it is a leaf).
A feasible solution consists of an assignment $\tilde{\varphi}:\tilde{C}\to \tilde{F}$ of each $c\in \tilde{C}$ to some $f\in \tilde{F}_c$. The cost of this solution is $\sum_{f\in \tilde{F}}h(\tilde{\varphi}^{-1}(f))$, where $h(\cdot)$ is a monotone submodular function over $\tilde{C}$ with $h(\emptyset)=0$. Similarly to SFL, we also express a feasible solution as $S=(S^f)_{f\in \tilde{F}}$, where $S^f=\tilde{\varphi}^{-1}(f)$, and let $\cost_{\DLA}(S)=\sum_{f\in \tilde{F}}h(S^f)$ be the associated cost. We define a convex-programming (CP) relaxation for DLA as follows:
\begin{align*}\label{eq:DLA-LP}\tag{DLA-CP}
        \min & \sum_{f \in \tilde{F}} \hat h(z^f) \\
        \text{s.t.} \quad
        &\sum_{f \in \tilde{F}_c} z_c^f = 1 \quad \forall c\in \tilde{C}. \\
        & z_c^f \geq 0 \quad \forall c\in \tilde{C},\ \forall f \in \tilde{F}.
\end{align*}
In a $0$-$1$ integral solution we interpret $z^f_c=1$ as $c$ being assigned to $f$. Recall that $\hat h(\cdot)$ is convex, which makes \eqref{eq:DLA-LP} a convex program. We also notice that each feasible assignment $S=(S^f)_{f\in \tilde{F}}$ corresponds to a feasible integral solution $z=(z^f)_{f\in \tilde{F}}$ to \eqref{eq:DLA-LP} with $\cost_{\DLA}(S)=\cost_{\DLA}(z)\coloneqq \sum_{f\in \tilde{F}}\hat{h}(z^f)$ and vice versa. Hence indeed \eqref{eq:DLA-LP} is a CP-relaxation of DLA.

The next lemma provides the claimed reduction from SFL on HST-type instances to DLA.
\begin{lemma}\label{lem:DLAreduction}
  Given a polynomial-time $O(\log D )$-approximate CP-rounding algorithm for DLA w.r.t.\ \eqref{eq:DLA-LP}, where $D$ is the depth of the tree, there is polynomial-time $O(\log\log N)$-approximate LP-rounding algorithm for SFL on HST-type instances with tree-depth $O(\log N)$ w.r.t.\ \eqref{eq:configuration-lp}.
\end{lemma}
\begin{proof}
Let $(C\cup F,d^T,g(\cdot))$ be the considered instance of SFL over an HST $T$, and $x$ be an input feasible fractional solution to \eqref{eq:configuration-lp} for this instance.

We build an instance of DLA as follows.
First, let $y^f_c\coloneqq \sum_{R\subseteq C:c\in R}x^f_R$: intuitively this is the fractional amount by which $c$ is assigned to $f$ in $x$. We set $h(\cdot)=g(\cdot)$ and $\tilde{T}=T$. Notice that $D=O(\log N)$. We set $\tilde{F}=F$ and map each $f\in \tilde{F}$ to the corresponding leaf of $T$ containing $f$. We also set $\tilde{C}=C$, and map each $c\in \tilde{C}$ to a node $v(c)$ as follows. Let $T_v$ be the subtree rooted at $v$ (containing $v$ and all its descendants) and $F_v$ be the facilities located in the leaves of $T_v$. Let also $\ell(c)$ be the leaf of $T$ containing $c$ in the mapping associated with $T$. We define $v(c)$ as the lowest ancestor of $\ell(c)$ such that $\sum_{f\in F_{v(c)}}y_c^f\geq 1/2$. Notice that $v(c)=\ell(c)$ is possible (in which case there is at least one facility $f$ colocated with $c$ at $l(c)$).

We next define a feasible fractional solution $z$ for \eqref{eq:DLA-LP} w.r.t this DLA instance as follows.
For each $c\in \tilde{C}$ we set $z^f_{v(c)}=y^f_c/(\sum_{f'\in F_{v(c)}}y^{f'}_c)$ if $f\in F_{v(c)}$, and otherwise $z^f_{v(c)}=0$. Let $\tilde{\varphi}$ be a solution to the DLA instance obtained with the CP-rounding algorithm in the claim w.r.t.\ $z$. We obtain a feasible solution $\varphi$ for the input instance by simply setting $\varphi(c)= \tilde{\varphi}(c)$.

It remains to analyze the cost of $\varphi$.
Define $\bar{z}^f_{v(c)}=y^f_c/(\sum_{f'\in F_{v(c)}}y^{f'}_c)$ for all $f\in F$. Notice that $\bar{z}\geq z$. By the definition of $\hat h(\cdot)$ and its monotonicity, $\hat{h}(z^f)\leq \hat{h}(\bar{z}^f)=\hat{h}(y^f/(\sum_{f'\in F_{v(c)}}y^{f'}_c))\leq 2\hat{h}(y^f)=2\hat{g}(y^f)$.
Notice that by plugging in $x_R^f$ for $\mu_R$ in the set in \eqref{eq:lovasz-extension} and by how $y$ is defined w.r.t.\ $x$ above, we get $\hat g(y^f) \leq \sum_{R \subseteq C}g(R)\cdot x_R^f$ and in particular $\sum_{f\in F}\hat g(y^f) \leq \open(x)$. Thus, we have $\cost_{\DLA}(z)\leq 2\open(x)$ and
\begin{equation}\label{eq:dlareduction-openingcosts}
  \open(\varphi)=\cost_{\DLA}(\tilde{\varphi})=O(\log D)\cdot\cost_{\DLA}(z)\leq O(\log \log N)\cdot 2\open(x).
\end{equation}
Consider next the connection cost. For each client $c\in C$, let $\Delta$ be the weight of the edge between $v(c)$ and its children in $T$. Observe that the distance between $v(c)$ and the leaves in $T_{v(c)}$ is exactly $2\Delta-1$.
Furthermore, both $c$ and $\varphi(c)$ are located in the leaves of $T_{v(c)}$ in the HST mapping. Hence
$
  d^T(c,\varphi(c))\leq 2(2\Delta-1).
$
We next compare the latter cost with the connection cost associated with $c$ in $x$, namely $\sum_{f\in F}d(c,f)y^f_c$. Let $\alpha=\sum_{f\in F_{v(c)}}y^f_c\geq \frac{1}{2}$. Suppose first that $\alpha\leq \frac{5}{6}$. Notice that for each $f\in F\setminus F_{v(c)}$, if any, the $f$-$c$ path in $T$ uses the edge, of weight $2\Delta$, between $v(c)$ and its parent. In particular, all such paths have length at least $2(4\Delta-1)$. Furthermore, $\sum_{f\in F\setminus F_{v(c)}}y^f_c=1-\alpha\geq \frac{1}{6}$. Thus
$$
  \sum_{f\in F}d(c,f)y^f_c\geq \sum_{f\in F\setminus F_{v(c)}}d(c,f)y^f_c\geq \frac{1}{6}2(4\Delta-1)\geq \frac{2}{3}(2\Delta-1).
$$
Consider next the complementary case where $\alpha\geq \frac{5}{6}$. Let $w(c)$ be the child of $v(c)$ along the $v(c)$-$c$ path in $T$. By the definition of $v(c)$, it must be the case that $\sum_{f\in F_{w(c)}}y^f_c< \frac{1}{2}$, and consequently $\sum_{f\in F_{v(c)}\setminus F_{w(c)}}y^f_c\geq \frac{1}{3}$. For each $f\in F_{v(c)}\setminus F_{w(c)}$, the $f$-$c$ path in $T$ has length exactly $2(2\Delta-1)$. Thus
$$
  \sum_{f\in F}d(c,f)y^f_c\geq \sum_{f\in F_{v(c)}\setminus F_{w(c)}}d(c,f)y^f_c\geq \frac{2}{3}(2\Delta-1).
$$
In both cases the connection cost of $c$ in $\varphi$ is at most $3$ times its connection cost in $x$. We conclude that $\conn(\varphi)\leq 3\conn(x)$. Altogether
$
  \cost(\varphi)\leq 3\conn(x)+O(\log\log N)\cdot 2\open(x)\leq O(\log\log N)\cdot \cost(x).
$
\end{proof}

\subsection{An Approximation Algorithm for DLA}\label{sec:approximationDLA}

In this section, we present a CP-rounding algorithm for DLA. Lemma \ref{lem:treeLProunding} follows by chaining Lemmas \ref{lem:DLAreduction} and \ref{lem:DLArounding}.
\begin{lemma}\label{lem:DLArounding}
  Given a feasible fractional solution $z$ to \eqref{eq:DLA-LP} on an instance of DLA with tree-depth $D$, in polynomial time one can compute a feasible (integral) solution to the same instance of cost at most $O(\log D)\cdot \cost_{\DLA}(z)$.
\end{lemma}
The CP-rounding algorithm from \cref{lem:DLArounding} is essentially the algorithm by Bosman and Olver \cite{BO20} with minor modifications that we introduced to simplify our correctness analysis. Also, the analysis of its approximation ratio is essentially identical to \cite{BO20}, but we reproduce it for the sake of completeness. In particular, we will exploit the following definitions and lemma from \cite{BO20}. Let $h:2^{\tilde{C}}\to \mathbb{R}_{\geq 0}$ be a monotone submodular function with $h(\emptyset)=0$. For a given $f\in \tilde{F}$ and a (possibly infeasible) solution $z$ to \eqref{eq:DLA-LP}, let $L_\theta(z^f)\coloneqq \{c\in \tilde{C}: z^f_c\geq\theta\}$ be the set of clients that are served fractionally by at least some value $\theta$ by $f$. Let also $z^{f|\theta}$ be obtained from $z^f$ by rounding down to $\theta$ the values larger than $\theta$, i.e. $z_c^{f|\theta}\coloneqq \min\{z^f_c,\theta\}$ for each $c\in \tilde{C}$. Given $\theta\in [0,1]$ and $z^f\in [0, 1]^{\tilde{C}}$, we say that the set $L_\theta(z^f)$ is $\alpha$-supported (w.r.t.\ $h$) if
$
  \hat h(z^f) - \hat h(z^{f|\theta}) \geq \alpha h(L_\theta(z^f)).
$
\begin{lemma}[{\cite[Lemma 5.2]{BO20}}]\label{lem:BO20}
  Given $z^f\in [0, 1]^{\tilde{C}}$ and $\alpha\in (0, 1]$, at least one of the following holds:
  (1) there exists $\theta\in [0, 1]$, which can be computed in polynomial time, such that $L_\theta(z^f)$ is $\frac{\alpha}{32}$-supported; (2) $2^{1/\alpha}h(L_1(z^f))\leq\hat h(z^f)$.
\end{lemma}

Our algorithm is Algorithm~\ref{alg:tree} in the figure. Recall that $\tilde{T}_v$ is the subtree rooted at node $v$, where $\tilde{T}_v$ includes $v$ and all its descendants. Furthermore, $\tilde{F}_v$ is the set of facilities mapped into the leaves of $\tilde{T}_v$. As usual the level of a node is its hop-distance from the root.
\begin{algorithm}
    \caption{}
    \footnotesize
    \textbf{Input:} Feasible solution $z$ to \eqref{eq:DLA-LP}
    \begin{algorithmic}[1]
        \State $S^f\gets \emptyset$ for all $f\in F$
        \For{$i=0,\dots, D$}
            \State For every node $v$ at level $D-i$, choose an arbitrary $f_v\in \tilde{F}_v$ and set $z^{f_v}\gets\sum_{f'\in \tilde{F}_v}z^{f'}$ and $z^{f'}\gets 0$

            for all $f'\in \tilde{F}_v\setminus \{f_v\}$\label{alg:merging}
            \If{there exists $\theta\in [0,1]$ such that $L_\theta(z^{f_v})$ is $\frac{1}{32\log (D+1)}$-supported}\label{alg:condition}
                \State For an arbitrary such $\theta$, set $S^{f_v}\gets S^{f_v}\cup L_\theta(z^{f_v})$ and $z^{f_v}_c\gets 0$ for all $c\in L_\theta(z^{f_v})$ \label{alg:choose-theta}
            \Else
                \State Set $S^{f_v}\gets S^{f_v}\cup L_1(z^{f_v})$ and $z^{f_v}_c\gets 0$ for all $c\in L_1(z^{f_v})$\label{alg:buying-integrals}
            \EndIf
        \EndFor
        \State For every $c\in \tilde{C}$ choose $f\in\tilde{F}_c$ such that $c\in S^{f}$ and set $S^{f'}\gets S^{f'}\setminus\{c\}$ for all $f'\in \tilde{F}\setminus\{f\}$ \label{alg:pruning}
        \State \textbf{return} $(S^f)_{f\in \tilde{F}}$
    \end{algorithmic}
    \label{alg:tree}
\end{algorithm}

Clearly Algorithm~\ref{alg:tree} runs in polynomial time.
The next two lemmas analyze the correctness and the approximation ratio of Algorithm \ref{alg:tree}, hence proving Lemma \ref{lem:DLArounding}.
\begin{lemma}\label{lem:DLAcorrectness}
  Algorithm \ref{alg:tree} computes a feasible DLA solution.
\end{lemma}
\begin{proof}
Consider a given client $c\in \tilde{C}$ such that $v(c)$ is at level $D-i$ in $\tilde{T}$. Let us show that the following invariant holds at the beginning of each iteration $j\leq i$: either $\sum_{f\in \tilde{F}_c}z^f_c=1$ or $c\in S^f$ for some $f\in \tilde{F}_c$. The invariant trivially holds for $j=0$. Assume that it holds up to the beginning of iteration $j<i$, and consider what happens during that iteration. Notice that for every node $v$ at level $D-j>D-i$, we either have that every $f \in\tilde{F}_v$ is a descendant of $v(c)$ or every $f \in\tilde{F}_v$ is not in $\tilde{F}_c$. Therefore, in Step \eqref{alg:merging} the value of $\sum_{f\in \tilde{F}_c}z^f_c$ does not change. In more detail, it remains $1$ by inductive hypothesis. The same value can decrease in Steps \eqref{alg:choose-theta} or \eqref{alg:buying-integrals}, however, this can only happen if $c$ is added to $S^{f_v}$ for some $f_v\in \tilde{F}_c$. Thus the invariant holds at the end of the $j$-th iteration, hence at the beginning of the next iteration $j+1$.

Due to the invariant, during the iteration $i$, when one considers the node $v=v(c)$, one has that either $c$ already belongs to some $S^f$ with $f\in \tilde{F}_c$, or $\sum_{f\in \tilde{F}_c}z^f_c=1$. In the latter case, after Step~\eqref{alg:merging}, $z_c^{f_{v}}=1$ where $f_{v}\in \tilde{F}_c$, so $c$ belongs to every set $L_\theta(z^{f_{v}})$ with $\theta\in [0,1]$. As a consequence, $c$ is added to $S^{f_v}$ either in Step \eqref{alg:choose-theta} or in Step \eqref{alg:buying-integrals}.

It might happen that a client $c$ is assigned \emph{also} to a facility not in $\tilde{F}_c$. Step~\eqref{alg:pruning} guarantees that the final assignment of $c$ is correct and unique.
\end{proof}

\begin{lemma}\label{lem:alg_cost}
  Algorithm~\ref{alg:tree} outputs a solution of cost at most $O(\log D)\cdot \cost_{\DLA}(z)$.
\end{lemma}
\begin{proof}
Recall that $\cost_{\DLA}(z)=\sum_{f\in\tilde{F}} \hat h(z^f)$.
We start by observing that the value of $\cost_{\DLA}(z)$ can not increase over time when $z$ changes during the execution of the algorithm. Indeed, Steps \eqref{alg:choose-theta} and \eqref{alg:buying-integrals} can only decrease the entries of $z$, hence $\cost_{\DLA}(z)$ by the monotonicity of $\hat{h}(\cdot)$. The only other changes of $z$ happen in Step \eqref{alg:merging}. Let us interpret this step as iteratively decreasing to zero $z^{f'}$ for each $f'\in \tilde{F}_v\setminus \{f_v\}$ and increasing $z^{f_v}$ by the same amount. The decrease of the cost at each step is $\hat{h}(z^{f_v})+\hat{h}(z^{f'})-\hat{h}(z^{f_v}+z^{f'})$. By the alternative definition of $\hat{h}(\cdot)$ as in \eqref{eq:lovasz-extension-bis} and its convexity, one has
$
  \hat{h}(z^{f_v}+z^{f'})=2\hat{h}\left(\frac{z^{f_v}+z^{f'}}{2}\right)\leq 2\left(\frac{1}{2}\hat{h}(z^{f_v})+\frac{1}{2}\hat{h}(z^{f'})\right)= \hat{h}(z^{f_v})+\hat{h}(z^{f'}).
$
Hence the decrease of the cost is non-negative as required.

For each facility $f$ and level $i$, let $\Delta^{\theta}_i (f)$ be the clients added to $S^f$ in Step \eqref{alg:choose-theta} during iteration $i$ (possibly $\Delta^{\theta}_i(f)=\emptyset$). We define similarly $\Delta^{1}_i (f)$ w.r.t.\ Step \eqref{alg:buying-integrals}. Notice that, by the submodularity (hence subadditivity) of $h(\cdot)$, the increase of the cost of the solution due to adding $\Delta$ to $S^f$ is at most $h(\Delta)$. Therefore we can upper bound the cost of the final solution $S=(S^f)_{f\in \tilde{F}}$ by
$$
  \cost_{\DLA}(S)\coloneqq \sum_{f\in \tilde{F}}h(S^f)\leq \sum_{i=0}^{D}\sum_{f\in \tilde{F}}\Big(h(\Delta^{\theta}_i (f))+h(\Delta^{1}_i (f))\Big).
$$
Let us upper bound the right-hand side of the above inequality. Let $z(i)$ denote the value of $z$ at the beginning of iteration $i$. From the previous observation, we have $\hat{h}(z(i))\leq \hat{h}(z)$ for every $i$. By Lemma \ref{lem:BO20} with $\alpha=\frac{1}{\log (D+1)}$, for any $\Delta^1_i(f)$ one has $h(\Delta^1_i(f))\leq \frac{1}{D+1}\hat{h}(z^f(i))$. Thus
\begin{equation}\label{eq:bound-delta-one}
  \sum_{i=0}^{D}\sum_{f\in \tilde{F}}h\big(\Delta^{1}_i (f)\big)\leq
  \sum_{i=0}^{D}\sum_{f\in \tilde{F}}\frac{1}{D+1}\hat{h}\big(z^f(i)\big)\leq
  \sum_{i=0}^{D}\frac{1}{D+1}\cost_{\DLA}(z(i))\leq
  \cost_{\DLA}(z).
\end{equation}

Let $z(D+1)$ be the value of $z$ at the end of the $D$-th iteration, hence in particular $\cost_{\DLA}(z(D+1))\geq 0$. Notice that $z=z(0)$. We can lower bound $\cost_{\DLA}(z)$ by
$$
  \cost_{\DLA}(z)\geq\sum_{i=0}^{D}\Big(\cost_{\DLA}(z(i))-\cost_{\DLA}(z(i+1))\Big).
$$
Let $z_1(i)$ be the value of $z$ obtained from $z(i)$ after applying Step \eqref{alg:merging} for all nodes of level $D-i$. Let also $z_2(i)$ be the value obtained from $z_1(i)$ if, for all the facilities $F'_i$ where Step \eqref{alg:choose-theta} is applied during iteration $i$, instead of setting $z^f_c=0$ one sets $z^f_c=\theta$ for the corresponding value of $\theta$. For the facilities not in $F'_i$ we simply let $z_2^f(i)=z_1^f(i)$. Observe that $z(i+1)\leq z_2(i)\leq z_1(i)\leq z(i)$. One has
\begin{align*}
  & \cost_{\DLA}(z(i))-\cost_{\DLA}(z(i+1))\\
  \geq & \cost_{\DLA}(z_1(i))-\cost_{\DLA}(z(i+1)) \geq
  \cost_{\DLA}(z_1(i))-\cost_{\DLA}(z_2(i))\\
  = & \sum_{f\in \tilde{F}}\hat{h}\Big(z_1^f(i)\Big)-\hat{h}\Big(z_2^f(i)\Big)=
  \sum_{f\in F'_i}\hat{h}\Big(z_1^f(i)\Big)-\hat{h}\Big(z_2^f(i)\Big)\geq
  \frac{\sum_{f\in F'_i}h\big(\Delta^{\theta}_i(f)\big)}{32\log(D+1)}=
  \frac{\sum_{f\in \tilde{F}}h\big(\Delta^{\theta}_i(f)\big)}{32\log(D+1)}.
\end{align*}
In the first two inequalities above we used the monotonicity of $\hat{h}(\cdot)$, while in the last inequality the definition of $\alpha$-supported. Altogether
\begin{equation}\label{eq:bound-delta-theta}
    \sum_{i=0}^{D}\sum_{f\in \tilde{F}}h\Big(\Delta^{\theta}_i (f)\Big)\leq 32\log(D+1)\cdot \sum_{i=0}^{D}\Big(\cost_{\DLA}(z(i))-\cost_{\DLA}(z(i+1)) \Big)\leq O(\log D)\cdot \cost_{\DLA}(z).
\end{equation}
By the monotonicity of $h(\cdot)$, Step~\eqref{alg:pruning} cannot increase the cost of the solution, hence the claim.
\end{proof}

\section*{Acknowledgements}

Fateme Abbasi and Jaros\l{}aw Byrka were supported by Polish National Science Centre (NCN) Grant 2020/39/B/ST6/01641.
Marek Adamczyk was supported by Polish National Science Centre (NCN) Grant 2019/35/D/ST6/03060.
Miguel Bosch Calvo, Fabrizio Grandoni, Krzysztof Sornat and Antoine Tinguely were supported by the SNSF Grant 200021\_200731/1.

\bibliographystyle{alpha}
\bibliography{bib}

\appendix

\section{Some Omitted Proofs about SFL}\label{sec:omitted}

Here we collect some proofs about SFL which were omitted in the main body.

\begin{proof}[Proof of Lemma \ref{lem:solveConfLP}]
Considering the dual of \eqref{eq:configuration-lp}:
\begin{align}\label{eq:configuration-lp-dual}\tag{Conf-DLP}
        \max \Big\{\sum_{c \in C} \alpha_c+ \sum_{f \in F} \beta_f :\ \sum_{c\in R} \alpha_c + \beta_f \leq g(R) + \sum_{c\in R} d(c,f),\,\forall R \subseteq C,\ \forall f \in F\Big\}.
\end{align}
Notice that for fixed $\alpha$ and $\beta$, the functions $g_f(R)\coloneqq g(R) + \sum_{c\in R} d(c,f) - \sum_{c\in R} \alpha_c - \beta_f$ are submodular. Thus, a call of a separation oracle on \eqref{eq:configuration-lp-dual} is equivalent to a minimization of all functions $g_f(\cdot)$, which can be done using polynomially many oracle calls of $g(\cdot)$ \cite{iwata2001combinatorial}. Therefore, an optimal primal solution with $\poly(N)$ many non-zero variables for \eqref{eq:configuration-lp} can be found in polynomial time \cite[Corollary 14.1g(v)]{schrijver1998theory}.
\end{proof}

\begin{proof}[Proof of Lemma \ref{lem:reduceNumberFacilities}]
Let $I=(C,F,d,g(\cdot))$ be the considered instance of SFL. Consider the complete weighted graph on nodes $C\cup F$, with weights induced by $d$. For each client $c$, let $f(c)$ be the facility closest to $c$. We create a dummy facility $f'(c)$ and add a dummy edge $\{c,f'(c)\}$ of weight $d(c,f(c))$. Let $F'$ be the set of newly created facilities. Observe that $|F'|=n$. Finally we remove $F$ and consider the metric $d'$ over $C\cup F'$ induced by the distances over the resulting graph. Let $I'=(C,F',d',g(\cdot))$ be the obtained instance of SFL. Given a solution $\varphi'$ for $I'$, we obtain a solution $\varphi$ for $I$ by simply assigning to $f(c)$ each client $c'$ assigned to $f'(c)$ in $\varphi'$.

Let us analyze the approximation factor introduced by this reduction. We first observe that $\cost_{I}(\varphi)\leq \cost_{I'}(\varphi')$. Indeed, $\open_{I}(\varphi)= \open_{I'}(\varphi')$. Furthermore, for each each client $c'$ assigned to $f'(c)$ by $\varphi'$, the associated connection cost w.r.t.\ $I$ is $d(c',f(c))\leq d(c',c)+d(c,f(c))=d'(c',f'(c))$. Hence $\conn_{I}(\varphi)\leq \conn_{I'}(\varphi')$.

Next consider an optimal solution $\opt$ for $I$. For each facility $f$ with $\opt^{-1}(f)\neq \emptyset$, let $c\in \opt^{-1}(f)$ be the client closest to $f$. We define a solution $\opt'$ for $I'$ by assigning all the clients in $\opt^{-1}(f)$ to $f'(c)$. Again, $\open_{I}(\varphi)= \open_{I'}(\varphi')$. For each client $c'$ assigned to $f$ in $\opt$, its connection cost in $I'$ is 
$$
  d'(c',f'(c))= d(c,c')+d(c,f(c))\leq d(c',f)+d(c,f)+d(c,f(c))\leq d(c',f)+2d(c,f)\leq 3d(c',f).
$$
Hence $\conn_{I'}(\opt')\leq 3\conn_{I}(\opt)$. The claim follows.
\end{proof}

\begin{proof}[Proof of Lemma \ref{lem:reduceDistances}]
Let us guess\footnote{Throughout this paper, by guessing we mean trying all the (polynomially many) possible options. Each such options leads to a different solution, and we return the best one.} the value $L=\max_{c\in C}d(c,\opt(c))$ for some optimal solution $\opt$. W.l.o.g. assume $L>0$, otherwise the problem is trivial. Consider the complete weighted graph on nodes $C\cup F$ with weights induced by $d$. Remove the edges of weight larger than $L$. We next compute a feasible solution in each connected component of the resulting graph separately. Notice that this part of the reduction is approximation preserving since no client can be assigned to a facility in a different connected component in $\opt$.

Let $C'$ and $F'$ be the clients and facilities, resp., in one such connected component $G'$, $X'=C'\cup F'$, and $d'$ be the metric induced by the distances in $G'$. Consider the corresponding SFL instance $I'=(C',F',d',g(\cdot))$. Notice that in each such instance $I'$ one has $d'_{\max}\leq NL$. We next change the location of elements of $X'$ as follows. We consider the ball $B(x)\coloneqq \{y\in X': d'(x,y)\leq \frac{\eps}{2n}L\}$ of radius $\frac{\eps}{2n}L$ around each $x\in X'$. Let $\cal{I}$ be a maximal (independent) set of such balls so that, if $B(x),B(y)\in {\cal I}$ for $x\neq y$, then $B(x)\cap B(y)=\emptyset$. For each $y$ with $B(y)\notin {\cal I}$, we consider any $B(x)\in {\cal I}$ with $B(x)\cap B(y)\neq \emptyset$ (which must exist since ${\cal I}$ is maximal) and colocate $y$ with $x$. Let $I''=\left(C',F',d'',g(\cdot)\right)$ be the resulting instance of SFL. Observe that $d''_{\max}\leq NL$ and $d''_{\min}\geq \frac{\eps}{n}L$.

Let $\tilde{I}$ be the union of all the instances $I''$, and $\tilde{d}$ be the associated distances (where inter-component distances can be considered to be $+\infty$). Given a solution $\varphi$ for $\tilde{I}$ (obtained by the union of all the solutions obtained for each instance $I''$), we return exactly the same solution $\varphi$ for $I$.

Let us analyze the approximation factor. Notice that $\open_I(\varphi)=\open_{\tilde{I}}(\varphi)$. Furthermore, for each client $c$, $d(c,\varphi(c))\leq \tilde{d}(c,\varphi(c))+\frac{2\eps}{n}L$, where in the latter term we consider the fact that each client and facility is moved at most at distance $\frac{\eps}{n}L$ from the original location. Hence $\conn_{I}(\varphi)\leq \conn_{\tilde{I}}(\varphi)+2\eps L$. Given an optimum solution $\opt$ for $I$, by a symmetric argument one has $\cost_{\tilde{I}}(\opt)\leq \cost_{I}(\opt)+2\eps L\leq (1+2\eps)\cost_{I}(\opt)$, where we used the fact that $\cost_I(\opt)\geq L$. Altogether an $\alpha\geq 1$ approximation algorithm for each instance $I''$ implies an $\alpha(1+2\eps)+2\eps\leq \alpha(1+4\eps)$ approximation for $I$.

Finally, we scale the distance $d''$ and $g(\cdot)$ by the same factor $\frac{2n}{\eps L}$ so that $d''_{\min}=2$ and $d''_{\max}\leq \frac{2nN}{\eps}$. Clearly this final scaling is approximation preserving.
\end{proof}

\section{Generalizations of SFL}\label{sec:extensions}

In this section we discuss some generalizations of SFL.

\subsection{Reduction of the Number of Facilities}\label{sec:reductionNumberFacilitiesGeneral}

In this section we consider the generalization of SFL, next called {\sc Affine} SFL, where the opening cost of each facility $f$ with assigned clients $R\neq \emptyset$ is $g_f(R)\coloneqq p_f+w_f\cdot g(R)$, where $p_f,w_f\geq 0$ are input values. Notice that this generalizes SFL {\sc with Additive} (resp., {\sc Multiplicative}) {\sc Opening Costs}. We also observe that each $g_f(\cdot)$ is non-negative monotone submodular.

We show how to reduce to the case where $m=\poly(n)$ (hence $N=\poly(n)$) while loosing a constant factor in the approximation. We will use this reduction in the following sections to convert an $O(\log\log N)$ approximation into an $O(\log\log n)$ one.

\begin{lemma}\label{lem:reduceNumFacilitiesAffine}
  For any constant $\eps>0$, there is a $(3+37\eps)$-approximate reduction from {\sc Affine} SFL to the special case where the number of facilities is $O_\eps(n^3)$.
\end{lemma}
\begin{proof}
First of all, consider the case $m\geq 2^n$. In this case we can solve the problem optimally in polynomial time via the following reduction to the {\sc Weighted Set Cover} problem.
For an instance $I=(C, F, d, g(\cdot))$ of \textsc{Affine SFL}, consider the instance $J=(\callu,\callr, \kappa)$ of {\sc Weighted Set Cover} with universe $\callu = C$, set collection $\callr = 2^C$ and weight function $\kappa$ given as $\kappa_R=0$ if $R=\emptyset$ and $\kappa_R=\min_{f\in F}(p_f+w_f\cdot g(R)+\sum_{c\in R}d(c,f))$ for $R \in 2^C\setminus\{\emptyset\}$ (which can be computed in $\poly(N)$ time).
Notice that $2^{|\callu|} = 2^n$ which is polynomially bounded in the input size of $I$.
The optimal solution to $J$ induces a solution of exactly the same cost to $I$ and vice versa. There is a simple dynamic program which solves {\sc Weighted Set Cover} in time $O(2^{|U|}\cdot|U|\cdot |{\cal R}|)$ \cite[Lemma~2]{fomin2004exact}. Applying this algorithm to $J$, one obtains an optimal solution for the input instance $I$ in time $O(2^n \cdot \poly(n,m))$, which is polynomial in $m$.

Hence it remains to consider the case $m\leq 2^n$. We show how to reduce the number of facilities to $O_\eps(n^2 \log (nN))=O_{\eps}(n^3)$, while losing the approximation factor in the claim. By exactly the same reduction as in Lemma \ref{lem:reduceDistances}, we can assume that in the input metric $d$ the maximum distance is $0<d_{\max}\leq NL$ and the minimum non-zero distance is $d_{\min}\geq \frac{\eps}{n}L$ while loosing a factor $(1+4\eps)$ in the approximation. Here $L$ is some value that lower bounds the cost of a given optimum solution $\opt$.
Let us guess the largest value $P$ of $p_f$ over the facilities with at least one assigned client in $\opt$. We discard all the facilities $f$ with $p_f>P$. Now, assuming $P>0$, we replace each $p_f$ with the value $p'_f\coloneqq \lceil \frac{p_f \cdot n}{\eps P}\rceil \cdot \frac{\eps P}{n}$ ($p'_f=p_f$ for $P=0$).
Notice that this can only increase the cost of a given solution $\varphi$, however this increase is upper bounded by $n\cdot \frac{\eps P}{n}\leq \eps\cdot \cost_I(\opt)$, where $I$ is the input instance of the problem. Hence this reduction preserves the approximation guarantee up to a factor $1+\eps$. After this reduction, the set ${\cal P}'$ of different possible values of $p'_f$ has cardinality at most $\frac{n}{\eps}$.

Let $I=(C,F,d,p',w,g(\cdot))$ be the instance of {\sc Affine} SFL obtained after the above two reductions. Consider the complete edge-weighted graph on nodes $C\cup F$, with weights induced by $d$. We modify this graph as follows. For each client $c$ and value $p'\in {\cal P}'$, we consider the set of facilities $F_{p'}$ with $p'_f=p'$. Let $F_{p'}(c,i)$, $i\geq 0$, be the facilities in $F_{p'}$ whose distances from $c$ are in the range $[\frac{\eps}{n}L\cdot (1+\eps)^i,\frac{\eps}{n}L\cdot (1+\eps)^{i+1})$. We also define the set $F_{p'}(c,-1)$ of the facilities in $F_{p'}$ at distance $0$ from $c$. Notice that there are at most $1+\lceil \log_{1+\eps}\frac{nN}{\eps}\rceil$ sets $F_{p'}(c,i)$ which are non-empty. For each $F_{p'}(c,i)\neq \emptyset$, we choose a facility $f=f_{p'}(c,i)$ with minimum value of $w_f$.
We create a dummy facility $f'=f'_{p'}(c,i)$ with opening cost $g'_{f'}(C')=p'+w_f\cdot g(C')$ for $C'\neq \emptyset$, and add a dummy edge $\{c,f'\}$ of weight $d(c,f)$. Let $F'$ be the set of dummy facilities.
Notice that, considering also the previous reduction, one has $|F'|\leq n\cdot \frac{n}{\eps}\cdot (1+\lceil \log_{1+\eps}\frac{nN}{\eps}\rceil)=O(n^2\log (nN))$. We remove the original facilities $F$, and let $d'$ be the metric given by the distances in the resulting graph $G'$ on nodes $C\cup F'$. We solve the problem on the resulting instance $I'=(C,F',d',p',w,g(\cdot))$. Given a solution $\varphi'$ for $I'$, we obtain a solution $\varphi$ for $I$ naturally as follows: if $\varphi'(c')=f'_{p'}(c,i)$, we assign $c'$ to $f_{p'}(c,i)$.

Let us analyze the approximation factor of this final reduction. The opening costs of $\varphi$ and $\varphi'$ are identical. Furthermore, for each client $c'$ assigned to $f=f_{p'}(c,i)$ in $\varphi$, and for $f'=f'_{p'}(c,i)$, one has
$$
  d(c',f)\leq d(c',c)+d(c,f)=d'(c',c)+d'(c,f')=d'(c',f').
$$
Hence $\cost_{I}(\varphi)=\cost_{I'}(\varphi')$.

Next consider an optimum solution $\opt$ for $I$. We construct a feasible solution $\opt'$ for $I'$ as follows. Let $S^f\neq \emptyset$ be the clients assigned to some $f\in F$ in $\opt$. Recall that the opening cost of $f$ is $g'_f(S^f)=p'_f+w_f\cdot g(S^f)$. Let $c\in S^f$ be the client at minimum distance $d(c,f)$ from $f$. Define $i$ as $-1$ if $d(c,f)=0$, and otherwise, $i$ such that $d(c,f)\in [\frac{\eps}{n}L\cdot (1+\eps)^i,\frac{\eps}{n}L\cdot (1+\eps)^{i+1})$. In $\opt'$ we reassign all the clients in $S^f$ to $f'=f'_{p'_f}(c,i)$. The opening cost associated with $f'$ in $\opt'$ is no larger than the corresponding cost in $\opt$ since
$$
  p'_{f'}+w_{f'}\cdot g(S^{f'})=p'_f+w_{f'}\cdot g(S^f)\leq p'_f+w_{f}\cdot g(S^f).
$$
In the last inequality above we used the fact that $f\in F_{p'_{f}}(c,i)$ and $f_{p'_{f}}(c,i)$ is the facility in the latter set with minimum $w_{f}$ value. The connection cost of each $c'\in S^f$ w.r.t.\ $\opt'$ satisfies
$$
  d'(c',f')=d'(c',c)+d'(c,f')= d(c,c')+d(c,f_{p'_f}(c,i))\leq d(c',f)+d(c,f)+(1+\eps)d(c,f)\leq (3+\eps)d(c',f).
$$
Altogether, $\cost_{I'}(\opt')\leq (3+\eps)\cost_{I}(\opt)$. Considering also the first two reductions, we obtain a global reduction which preserves the approximation guarantee up to a factor $(1+4\eps)(1+\eps)(3+\eps)\leq 3+37\eps$.
\end{proof}

\subsection{SFL with Multiplicative Opening Costs}

In this section we sketch the proof of Theorem \ref{thr:mainMultiplicative}. By Lemma \ref{lem:reduceNumFacilitiesAffine}, it is sufficient to provide an $O(\log\log N)$ approximation.

For $f\in F$ and $R \subseteq C$ let $g_f(R) \coloneqq w_f \cdot g(R)$.
Note that $g_f(\cdot)$ is submodular, monotone and has $g(\emptyset) = 0$ for every $f \in F$.
For any (partial) assignment $S=(S^f)$ and any vector $(x_R^f)_{R \subseteq C}^{f\in F}$ let also $\open'(S) \coloneqq \sum_{f\in F}g_f(S^f)$, resp.\ $\open'(x) \coloneqq \sum_{f\in F}\sum_{R \subseteq C} g_f(R) \cdot x_R^f$ and $\cost'(S) \coloneqq \open'(S)+\conn(S)$ resp.\ $\cost'(x) \coloneqq \open'(x)+\conn(x)$.

By these definitions, the LP-relaxation of the {\sc multSFL} is given by the constraints from \eqref{eq:configuration-lp} and the objective $\cost'(\cdot)$. In particular, the LP-relaxation of {\sc multSFL} can be solved with the approach from \cref{lem:solveConfLP}.
We keep the merging rule defined in \cref{sec:preliminaries} and the sampling procedure from \cref{sec:reducingConnection}. It is easy to verify that the vector $\ddot x$ resulting from this procedure fulfills \cref{lem:costResidualClients} w.r.t.\ $\open'$ instead of $\open$.

We reduce {\sc multSFL} to a similar problem to DLA which we call DLA$^*$ which is the same problem as DLA and with the same input variables as DLA, additional inputs $\tilde w_f \geq 0$ for every $f\in \tilde F$ and cost $\cost^*_{\DLA}(\varphi) = \sum_{f\in \tilde F} h_f(\varphi^{-1}(f))$ where $h_f(\cdot) \coloneqq \tilde w_f h(\cdot)$ for every $f \in \tilde F$.
Its convex relaxation is given by the constraints in \eqref{eq:DLA-LP} with the cost function $\cost^*_{\DLA}(z) \coloneqq \sum_{f \in \tilde F} \hat h_f(z^f)$ (where $\hat h_f$ is the Lov\'{a}sz extension of $h_f$).
The reduction described in \cref{lem:DLAreduction} can be reproduced to reduce {\sc multSFL} to DLA$^*$. We define the input values of DLA$^*$ w.r.t.\ {\sc multSFL} in the same way we define the input values of DLA w.r.t.\ SFL, with additionally $\tilde w_f = w_f$ for every $f\in F$.
Notice that $h_f(\cdot) = \tilde w_fh(\cdot) = g_f(\cdot) = w_f g(\cdot)$.
Every reasoning made in the proof of \cref{lem:DLAreduction} stays valid.

We now adjust Algorithm~\ref{alg:tree} for DLA$^*$ as follows: in Step~\ref{alg:merging}, we select the facility $f_v\in\tilde{F_v}$ with minimum weight $\tilde w_{f_v}$. In the if-clause~\ref{alg:condition}, we search and verify for supportedness w.r.t.\ $h_{f_v}$ instead of $h$ (which is equivalent unless $\tilde w_{f_v}=0$, in which case $L_\theta(z^{f_v})$ is supported for every $\theta$).
Since the new algorithm functions exactly like \cref{alg:tree}, except for an arbitrary selection step becoming determined (in particular, the new algorithm is a possible implementation of \cref{alg:tree}), its correctness is implied by the correctness of \cref{alg:tree}.

Notice that since $f_v$ in Step \ref{alg:merging} is now chosen to have minimal weight, we have for any $f' \in \tilde F_v\setminus \{f_v\}$
\begin{equation*}\label{eq:mult-value-decrease}
    \hat h_{f_v}\big(z^{f_v}+z^{f'}\big)
    \leq \hat h_{f_v}\big(z^{f_v}\big) + \hat h_{f_v}\big(z^{f'}\big)
    \leq \hat h_{f_v}\big(z^{f_v}\big) + \hat h_{f'}\big(z^{f'}\big),
\end{equation*}
which means that the cost of $z$ does not increase at any time by the arguments as before. Also, notice that since $h_f$ is submodular, monotone and $h_f(\emptyset)=0$ we can apply Lemma~\ref{lem:BO20} with respect to $h_{f_v}$ instead of $h$. Thus, the cost of the sets added at Step~\ref{alg:choose-theta} and Step~\ref{alg:buying-integrals} is still bounded as in \eqref{eq:bound-delta-one} and \eqref{eq:bound-delta-theta}.

\subsection{SFL with Additive Opening Costs}

In this section we sketch the proof of Theorem \ref{thr:mainAdditive}. As in the previous section, by Lemma \ref{lem:reduceNumFacilitiesAffine}, it is sufficient to provide an $O(\log\log N)$ approximation.

Similarly to the previous section, we define the set function $g_f(\cdot)$ as $g_f(R) = g(R) + p_f$ for $R\neq \emptyset$ and $g_f(\emptyset)=0$.
As argued in the previous section, we can find an optimum to the LP relaxation of {\sc addSFL} and reduce it to the problem
DLA$^*$ as defined in the last section, but with input weights $\tilde p_f$ instead of $\tilde w_f$ and $h_f(\cdot)$ as $h_f(R) \coloneqq h(R) + p_f$ for $R\neq \emptyset$, and $h_f(\emptyset) = 0$.

We adapt \cref{alg:tree} like in the previous section:
in Step~\ref{alg:merging}, we select the facility $f_v\in\tilde{F_v}$ with minimum weight $\tilde p_{f_v}$. In the if-clause~\ref{alg:condition}, we search and verify for supportedness w.r.t.\ $h_{f_v}$ instead of $h$.
The correctness of the new algorithm here is given by the same argument as in the previous section.
Notice that by \eqref{eq:lovasz-extension-bis} we have $\hat h_f(z) = \hat h(z) + p_f \cdot \max_{c\in \tilde C} z_c$, which implies $\hat h_{f_v}(z^{f_v}+z^{f'}) \leq \hat h_{f_v}(z^{f_v}) + \hat h_{f'}(z^{f'})$ with $f_v$ chosen as in Step~\ref{alg:merging} in \cref{alg:tree}. The cost of $z$ does therefore not increase throughout the algorithm.
Bounding the cost of sets added to the solution at Step~\ref{alg:choose-theta} and Step~\ref{alg:buying-integrals} can be done, like for {\sc multSFL}, by applying \cref{lem:BO20} to $h_{f_v}$.

\section{Universal Stochastic Facility Location}\label{sec:universal}

In this section we sketch our approximation algorithm for {\sc univFL}. We first present a weaker approximation factor $O(\log\log N+\log\log \frac{d_{\max}}{d_{\min}})$. Later we will show how to refine it.

Define $g(R)\coloneqq \prob_{A \sim \pi}[R \cap A \neq \emptyset]$. We observe that this function is monotone submodular and $g(\emptyset)=0$. Recall that $g(c)=g(\{c\})$ for every $c\in C$. W.l.o.g. we can assume $g(c)>0$ since otherwise we can discard $c$. We can define the objective function of {\sc univFL} for a given assignment $\varphi:C\to F$ as
$$
  \cost(\varphi)=\conn(\varphi)+\open(\varphi)=\sum_{c \in C} d(c,\varphi(c)) \cdot g(c) + \sum_{f \in F} w_f\cdot g(\varphi^{-1}(f)).
$$
Notice that only the connection cost changes w.r.t.\ {\sc multSFL}. In more detail, the connection cost of each client $c$ is scaled by the factor $g(c)$.

We can similarly define a configuration LP for {\sc univFL}, and solve it by the same arguments as in Lemma \ref{lem:solveConfLP}. We next use an analogous notation as for SFL. Let $\dot{x}$ be an optimal solution to this LP with $\poly(N)$ many non-zero variables. We can apply the first stage of our algorithm for SFL (described in \cref{sec:reducingConnection}) with essentially no changes. This will lead to a partial assignment $S_1$ of expected cost $\expected[\cost(S_1)]\leq \ln\ln N\cdot \cost(\dot{x})$ and serving the clients $C_1$, where $\prob[c\notin C_1]\leq \frac{1}{\ln N}$. Mapping the metric over an HST $T$ and considering the restriction $\ddot{x}$ of $\dot{x}$ to $C_2\coloneqq C\setminus C_1$, we obtain that $\expected[\cost_{HST}(\ddot{x})]=O(\cost(\dot{x}))$. A reduction similar to the one in Lemma \ref{lem:DLAreduction} works also in this case (since the scaling of the fractional solution is done on a per-client base). However in this case $D=O(\log \frac{d_{\max}}{d_{\min}})$ (since we did not reduce the ratio $\frac{d_{\max}}{d_{\min}}$ in a preprocessing step). Hence we can apply the result from Lemma \ref{lem:DLArounding} to obtain an assignment covering $C_2$ of expected cost $O(\log\log \frac{d_{\max}}{d_{\min}})\cdot \cost(\dot{x})$. This concludes the sketch of the $O(\log\log N+\log\log \frac{d_{\max}}{d_{\min}})$ approximation.

We next improve this bound via a preprocessing step. Recall that $0<\pi_{\min}\coloneqq \min_{c\in C}\{g(c)\}$. We first scale the ratio $d_{\max}/d_{\min}$. Let us guess the largest distance $L=\max_{c\in C}\{d(c,\opt(c))\}$ in some optimal (universal) solution $\opt$. Notice that $\cost(\opt)\geq \pi_{\min}L$. We use essentially the same arguments as in Lemma \ref{lem:reduceDistances}, we can enforce that $d_{\max}\leq NL$ and $d_{\min}\geq \frac{\eps}{n}\pi_{\min}L$. Hence we obtain $\frac{d_{\max}}{d_{\min}}\leq \frac{nN}{\eps \pi_{\min}}$.

Now let us reduce the number of facilities $m$ to $O(n+\log\frac{1}{\pi_{\min}})$ (hence $N$ as well).
Here we use essentially the same argument as in the proof of Lemma \ref{lem:reduceNumFacilitiesAffine} (with $p_f=0$). In more detail, we can assume that
$m\leq 2^n$. Indeed, otherwise we can reduce the input instance to a {\sc Weighted Set Cover} instance (that we can solve exactly in polynomial time) in the same way as in the mentioned lemma, with the difference that now, for $R\neq \emptyset$, we set $\kappa_R=\min_{f\in F}\{w_f\cdot g(R)+\sum_{c\in R}d(c,f)\cdot g(c)\}$. By the rest of the construction in the same lemma, we can reduce (with a constant loss in the approximation factor) our instance to one where there are $O(\log \frac{d_{\max}}{d_{\min}})=O(\log \frac{n2^n}{\eps \pi_{\min}})=O(n+\log \frac{1}{\pi_{\min}})$ facilities per client. Altogether we reduce $N$ to $N'=O(n(n+\log \frac{1}{\pi_{\min}}))$. Now we can apply again the above scaling trick over the distances (with $N$ replaced by $N'$) to obtain distances $d'$ which satisfy:
$$
  \frac{d'_{\max}}{d'_{\min}}\leq \frac{nN'}{\eps \pi_{\min}} = O\bigg(\frac{n^3+n^2\log \frac{1}{\pi_{\min}}}{\pi_{\min}}\bigg).
$$
This leads to the approximation factor
$$
  O\bigg(\log\log \frac{d'_{\max}}{d'_{\min}}+\log\log N'\bigg)=O\bigg(\log\log \frac{n}{\pi_{min}}\bigg).
$$

\section{Lower-Bound for a Greedy Algorithm}\label{app:lower-bound}

In this section we prove Theorem~\ref{thm:lower-bound}, i.e. that a natural greedy algorithm for SFL has an approximation ratio at least $\Omega(\log \log n)$.
We will first specify the greedy algorithm being analyzed, then give the construction of the instances, finally prove that the algorithm indeed behaves poorly on the instances provided.

\subsection{Greedy Algorithm}
We consider the most natural (set cover type) greedy algorithm that works as follows.
While not all clients are served by facilities, select (and include in the solution) a subset $R$ of still uncovered clients and a location of facility $f$ minimizing the following \emph{cost-effectiveness ratio}
\[
  \frac{g(R \cup T) - g(T) + \sum_{c \in R} d(c,f)}{|R|},
\]
where $T$ is the (possibly empty) set of clients already served by a facility in location $f$. Notice that $g(R \cup T) - g(T)$ is the facility cost increase resulting from adding clients from set $R$.

Such a natural greedy algorithm is known to be a $1.861$-approximation algorithm for FL~\cite{JMMSV03}, which is a special case of our setting in which $g(\cdot)$ is a constant function not depending on the set of clients being served (unless this is an empty set for which the opening cost is $0$).

The above description of the greedy algorithm does not specify how ties are broken, namely what to do if there is more than one minimizer of the cost-effectiveness ratio. In order to facilitate the presentation of our lower bound example, we will assume that ties are broken consistently through the following preference order:
\begin{enumerate}
    \item $R$ is a set of two clients from different locations and $f$ is located at non-zero distance from each of the clients from $R$;
    \item $R$ contains a single client and $f$ is at the same location as the client;
    \item any other configuration.
\end{enumerate}
We call the algorithm specified above {\tt GreedySFL}.

\subsection{Instance Construction}
We will now describe a construction of instances that are difficult for {\tt GreedySFL}.
The instances are parameterized by an integer $\dime > 1$.
Our construction has $2^{\dime}$ locations, on which there are in total $\dime \cdot 2^{\dime}$ clients. There are $2^{\dime} + \dime \cdot 2^{\dime-1}$ locations for a possible facility, every location has the same facility opening cost being the function $g \colon 2^C \rightarrow \RR_{\geq 0}$ of the set of clients being served. The nontrivial part of the construction lies in the definition of the values $g(R)$ for all possible subsets of clients $R \subseteq C$.

We present the construction by first defining a particular structure behind the set of clients.
This is followed by the definition of the set of facilities and distances in the constructed instance of SFL.
Finally, we define the cost function $g$ and show its key properties.

\paragraph{Topology.}
Consider a $\dime$-dimensional hypercube, which will form the geometry of our instance. The set of vertices of the hypercube is $V = \{0,1\}^{[\dime]}$, where $[\dime] \coloneqq \{1,2,\dots,\dime\}$.
For $v \in V$, we write $v=(v_1,v_2,\dots,v_{\dime})$.
There are $\dime$ clients on each vertex, meaning that the set of clients is $C = V \times [\dime] = \{(v,l): \ v\in V,\ l \in [\dime]\}$.
Therefore, $|C| = \dime \cdot 2^{\dime}$.
For $(v,l) \in C$, we call $l$ an \emph{index of $c$}.

We will consider the operation of \emph{activating selected dimensions} $A \subseteq [\dime]$, which intuitively has two effects:
\begin{enumerate}
    \item It flattens the hypercube in these dimensions, making vertices that originally differed only in dimensions from $A$ indistinguishable;
    \item It activates all clients $(v,l)$ with $l \in A$.
\end{enumerate}
Formally, for a subset of dimension $A \subseteq [\dime]$ and for $v \in V$ we define $v^A=(v_1^A, v_2^A,\dots,v_{\dime}^A)$ as
\begin{center}
    $v_i^A \coloneqq
      \begin{cases}
        *, \quad & \text{for} \ i \in A; \\
        v_i, \quad& \text{for} \ i \notin A. \\
      \end{cases}
    $
\end{center}
Then, for $A \subseteq [\dime] $ and for $U \subseteq V$ we define $U^A \coloneqq \{v^A: v \in U\}$ as \emph{a collapsed set of vertices}.
Next, for $A \subseteq [\dime] $ and for $R \subseteq C$ we define $R^A \coloneqq \{(v^A,l): (v,l) \in R \}$ as \emph{a collapsed set of clients}. Define $U(R,A) \coloneqq \{v^A \in U^A : \exists l \in A \text{ such that } (v^A,l)\in R^A\}$ as \emph{the set of collapsed vertices containing an activated client from collapsed $R$} (for an example see Figure~\ref{fig1}).

\newcommand{\tikzI}{
\begin{tikzpicture}
    \node[shape=circle,draw=black,label={north west: $\tau_{000}$}] (A) at (0,0) {\textcolor{white}{$v_3$}};
    \node[shape=circle,draw=black ,label={north west: $\tau_{100}$}] (D) at (3.3,0) {\textcolor{white}{$v_3$}};
    \node[shape=circle,draw=black, label={north west: $\tau_{001}$}] (E) at (1.1,1.1) {\textcolor{white}{$v_3$}};
    \node[shape=circle,draw=black, label={north west: $\tau_{101}$}] (F) at (4.4,1.1) {\textcolor{white}{$v_3$}};
    \node[shape=circle,draw=black, label={north west: $\tau_{010}$}] (B) at (0,2.75) {\textcolor{white}{$v_3$}};
    \node[shape=circle,draw=black, label={north west: $\tau_{110}$}] (C) at (3.3,2.75) {\textcolor{white}{$v_3$}};
    \node[shape=circle,draw=black, label={north west: $\tau_{011}$}] (H) at (1.1,3.85) {\textcolor{white}{$v_3$}};
    \node[shape=circle,draw=black, label={north west: $\tau_{111}$}] (G) at (4.4,3.85) {\textcolor{white}{$v_3$}};
    \path [-](A) edge (B);
    \path [-](B) edge (C);
    \path [-](A) edge (D);
    \path [-](D) edge (C);
    \path [-](A) edge (E);
    \path [-](E) edge (H);
    \path [-](G) edge (F);
    \path [-](D) edge (F);
    \path [-](G) edge (H);
    \path [-](E) edge (F);
    \path [-](B) edge (H);
    \path [-](C) edge (G);
\end{tikzpicture}
}

\newcommand{\tikzII}{
\begin{tikzpicture}
    \node (A) at (0, 0) [shape=circle,draw=black,label={north west: $\tau_{000}$}] {
    \shortstack{$(\tau_{000},1)$,\\ $ (\tau_{000},2)$,\\$(\tau_{000},3)$}
    };
\end{tikzpicture}
}

\newcommand{\tikzIII}{
\begin{tikzpicture}
    \node[shape=circle,draw=black,label={[align=center]north west :$\tau_{0*0}$}] (A) at (0,0) {\textcolor{white}{$v_1$}};
    \node[shape=circle,draw=black, label={[align=center]north west :$\tau_{1*0}$}] (D) at (3.3,0) {\textcolor{white}{$v_2$}};
    \node[shape=circle,draw=black, label={[align=center]north west :$\tau_{0*1}$}] (E) at (1.1,1.65) {\textcolor{white}{$v_3$}};
    \node[shape=circle,draw=black, label={[align=center]north west :$\tau_{1*1}$}] (F) at (4.4,1.65) {\textcolor{white}{$v_4$}} ;
    \path [-](A) edge (D);
    \path [-](A) edge (E);
    \path [-](D) edge (F);
    \path [-](E) edge (F);
\end{tikzpicture}
}

\begin{figure}[ht]
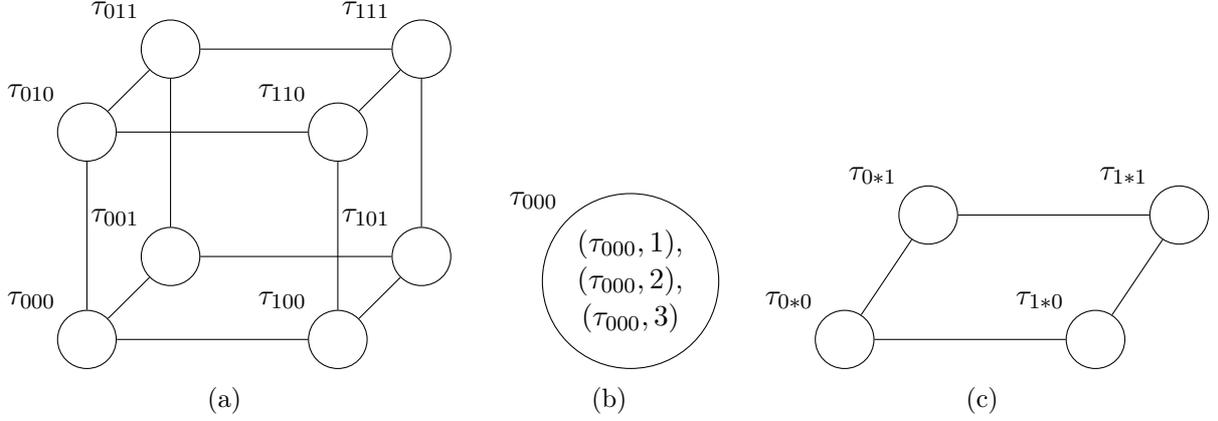

  \begin{subfigure}[b]{0.4\linewidth}
    \centering
    \tikzI
    \caption{}
    \label{fig:a}
    \vspace{4ex}
  \end{subfigure}
  \begin{subfigure}[b]{0.2\linewidth}
    \centering
    \tikzII
    \caption{}
    \label{fig:0b}
    \vspace{4ex}
  \end{subfigure}
  \begin{subfigure}[b]{0.38\linewidth}
    \centering
    \tikzIII
    \caption{}
    \label{fig:c}
    \vspace{4ex}
  \end{subfigure}
  \caption{An example of a 3-dimensional hypercube.
    For readability, we use the notation $\tau_{v_1 v_2 v_3} = (v_1,v_2,v_3) \in V$.
    Figure (a) shows the set of vertices $V=\{\tau_{000}, \dots, \tau_{111}\}$, where
    Figure (b) depicts three clients on the vertex $\tau_{000}$.
    Figure (c) shows the set of collapsed vertices for $A=\{2\}$, i.e. the set $U^{\{2\}}$.
    Suppose $A=\{2\}$ and $R=\{(\tau_{100},1), (\tau_{000},2), (\tau_{011},2), (\tau_{010},2), (\tau_{111},3)\}$.
    In this case $U(R,A)= \{\tau_{0*0}, \tau_{0*1}\}$ because the activated clients are $(\tau_{000},2), (\tau_{011},2), (\tau_{010},2)$ which, after collapsing, are contained in vertices $\tau_{0*0}, \tau_{0*1}, \tau_{0*0}$ respectively.
    }
  \label{fig1}
\end{figure}

\paragraph{Metric of the instance.}
Above, for a fixed dimensionality $\dime$, we have defined a set of $2^{\dime}$ locations such that $\dime$ clients reside in each location.
In what follows we will define the set of locations $F$ in which facilities may be open and the metric on $C \cup F$.

Consider the natural graph of the hypercube $G=(V,E)$ with vertices $V$ being the vertices of the hypercube and two vertices $v,u$ being connected by an edge $e = \{v,u\} \in E$ if their address in the hypercube differs in exactly one coordinate.
Define the set of facility locations as $F = V \cup E$, meaning that facilities may be open on vertices of the hypercube and on (midpoints of) edges connecting neighboring hypercube vertices.
Hence, the number of facility locations is $2^{\dime} + \dime \cdot 2^{\dime-1} \leq O(\dime \cdot 2^{\dime})$.

We now define the metric space over $C \cup F$. We start by setting the lengths of edges $e \in E$ of graph $G$. For an edge $e \in E$ connecting vertices that differ in the coordinate $i \in [\dime]$, we define the length of $e$ as $p_{i} = \frac{1}{2(\dime+1-i)}$.
For a facility $f$ located on the (midpoint of) edge $e = \{v,u\}$, define its distance to the endpoints of $e$ to be half the length of $e$, i.e., $d(v,f) = d(u,f) = \frac{p_i}{2}$. Finally, for any two points in $V \cup E$ we define their distance to be the shortest-path distance on $G$.

\paragraph{Submodular facility cost function.}
Now consider a random process $\pi$ which samples $A \subseteq [\dime]$ by independently sampling each $i \in [\dime]$ with probability $p_i$ (notice $p_i \in [0,1]$).
$p_A = \Pi_{i \in A} \:p_i$ denotes the probability that $A \subseteq [\dime]$ is sampled.
We define $g\colon 2^C \rightarrow \RR_{\geq 0}$ as
$$
  g(R) \coloneqq \expected_{A \sim \pi}\big[|U(R,A)|\big] = \sum_{A \subseteq [\dime]} p_A \cdot |U(R,A)|
$$
as the expected cardinality of collapsed vertices containing an activated client from collapsed $R$.
Notice that while $|U(R,A)|$ is a random variable $g(\cdot)$ is a deterministic function.
The reason of introducing the random process is its usefulness in a compact definition of $g(\cdot)$.
Moreover, it is also helpful in order to derive some important properties of $g(\cdot)$.

Observe that the above defined function $g(\cdot)$ is monotone and has $g(\emptyset)=0$.
To show that $g$ is submodular, it suffices to check that the function $R \mapsto |U(R,A)|$ is submodular for any fixed $A$. In our lower bound construction, we will exploit the behaviour of $g$ as a result of splitting a set of clients according to a fixed dimension $i \in [\dime]$.

\paragraph{Specific properties of the facility cost function.}
We will use the following notation: for every $i \in [\dime]$ and every $R \subseteq C$, let
$V_0^{i}=\{(v_1,v_2,\dots,v_{\dime})\in V:\ v_{i}=0\}$ and $V_1^{i}=\{(v_1,v_2,\dots,v_{\dime})\in V:\ v_{i}=1\}$;
and accordingly
$R_0^{i}= \{(v,l)\in R:\ v \in V_0^{i}\}$ and $R_1^{i} = \{(v,l) \in R:\ v \in V_1^{i}\}$.

First, we observe that for any client $(v,l)$, we have $g((v,l)) = p_l$.
The following lemma formally states the key property of $g$, namely that when considering two sets of clients $R_0^i$ and $R_1^i$, saving from serving a (merged) set $R_0^i \cup R_1^i$ is at most the connection cost increase caused by the merge.
Intuitively, it means that it is not strictly better to serve a merged set $R_0^i \cup R_1^i$ than to serve them separately.

\begin{lemma} \label{l:dim_split}
  We have $g(R_0^i)+g(R_1^i)-g(R_0^i \cup R_1^i) \leq \min\{|R_0^i|,|R_1^i|\} \cdot p_i.$
\end{lemma}
\begin{proof}
For every $A \subseteq [\dime]$ we have
\begin{align}
    &|U(R_0^i,A)|+|U(R_1^i,A)|-|U(R_0^i \cup R_1^i,A)|\nonumber\\
    \leq\:&
    |U(R_0^i,A)|+|U(R_1^i,A)|- \max\{|U(R_0^i,A)|, |U(R_1^i,A)|\}\nonumber\\
    \leq\:&
    \min\{|U(R_0^i,A)|, |U(R_1^i,A)|\}\nonumber\\
    \leq\:&
    \min\{|R_0^i|, |R_1^i|\}. \label{ineq:usr-split-i-in}
\end{align}
Furthermore, for every $A \subseteq [\dime]$ such that $i \notin A$ the collapsed halves of the cube are disjoint, i.e.\ $(R_0^i)^A \cap (R_1^i)^A = \emptyset$.
Moreover, $U(R_0^i,A) \cap U(R_1^i,A) = \emptyset$, hence $U(R_0^i,A) \cup U(R_1^i,A) = U(R_0^i \cup R_1^i,A)$.
Therefore, for every $A \subseteq [\dime] \setminus \{i\}$, we have
\begin{align}
    |U(R_0^i,A)| + |U(R_1^i,A)| - |U(R_0^i \cup R_1^i,A)| = 0.\label{ineq:usr-split-i-notin}
\end{align}
From the definition of $g(\cdot)$, we obtain
\begin{align*}
    g(R_0^i)+g(R_1^i)-g(R_0^i \cup R_1^i)
    =&
    \sum_{A \subseteq [\dime]} \prob[A] \cdot (|U(R_0^i,A)|+|U(R_1^i,A)|-|U(R_0^i \cup R_1^i, A)|)\\
    =&
    \sum_{A \subseteq [\dime]: i \in A} \prob[A] \cdot (|U(R_0^i,A)|+|U(R_1^i,A)|-|U(R_0^i \cup R_1^i, A)|)\\
    +&
    \sum_{A \subseteq [\dime]: i \notin A} \prob[A] \cdot (|U(R_0^i,A)|+|U(R_1^i,A)|-|U(R_0^i \cup R_1^i, A)|)\\
    \stackrel{\eqref{ineq:usr-split-i-in},\eqref{ineq:usr-split-i-notin}}{\leq}&
    \sum_{A \subseteq [\dime]: i \in A} \prob[A] \cdot \min\{|R_0^i|, |R_1^i|\}
    +
    \sum_{A \subseteq [\dime]: i \notin A} \prob[A] \cdot 0\\
    =&
    \quad p_i \cdot \min\{|R_0^i|, |R_1^i|\}.
\end{align*}
This concludes the proof of the lemma.
\end{proof}

The following lemma describes the behaviour of {\tt GreedySFL} on clients from a single vertex: the cost-effectiveness of serving one client is not worse than the cost-effectiveness of serving more clients from the same vertex.
\begin{lemma}\label{l:singleton}
  For every $R \subseteq C$ such that $|\{v \in V: (v,l) \in R\}|=1$, there exists some $c \in R$ satisfying
  \begin{equation}\label{eq:singleton-lemma}
    g(c)\leq \frac{g(R)}{|R|}.
  \end{equation}
\end{lemma}
\begin{proof}
 As all clients considered in this lemma are located on the same vertex, we omit them from the notation, i.e. we write $c = l$ instead of $c = (v,l)$.
 First, recall that for every $l \in [\dime]$ we have $g(l) = p_l$.
 Let $i_1 = \min(R)$, so $p_{i_1} = g(i_1)$ is equal to the minimizer of the left-hand side of \eqref{eq:singleton-lemma}.

 As all clients are located in the same vertex, we have $|U(R,A)| = \indic[R \cap A \neq \emptyset]$, where $\indic[X]$ is the indicator of $X$ being true.
 Therefore we have
 $$
   g(R) = \expected_{A \sim \pi}\big[|U(R,A)|\big] = \prob_{A \sim \pi}[R \cap A \neq \emptyset].
 $$
 Hence, in order to prove the lemma, it is enough to show the following inequality
 \begin{align}\label{ineq:singleton-better}
   p_{\min(R)} \cdot |R| \leq \prob_{A \sim \pi}[R \cap A \neq \emptyset].
 \end{align}
 We show this by induction on the cardinality of $R$.
 If $|R|=1$ then \eqref{ineq:singleton-better} is trivial (the inequality holds with equality).
 Suppose now the claim is true for $|R| \leq r_0$. We show that it also holds for $|R|=r_0+1$.

 Let $i_2$ be the second smallest element in $R$.
 Let $E_1$ be the event of element $i_1$ being selected for the random subset $A \sim \pi$ (hence, we will write $\prob$ as a shortcut for $\prob_{A \sim \pi}$).
 Let $E_2$ be the event of at least one element from $R \setminus \{i_1\}$ being selected to the set $A$.
 Note that $E_1$ and $E_2$ are independent events and that we already have an estimation of $\prob[E_2]$ from our inductive assumption, hence
 \begin{equation*}
  \begin{split}
   \prob[R \cap A \neq \emptyset] = \quad \prob&[E_1 \cup E_2] = \prob[E_1] + \prob[E_2] - \prob[E_1] \cdot \prob[E_2]\\
   = \quad p&_{i_1} + (1-p_{i_1}) \cdot \prob[E_2]\\
    \stackrel{\text{induc.}}{\geq} \: p&_{i_1} + (1-p_{i_1}) \cdot p_{i_2} \cdot (|R|-1)\\
    = \quad p&_{i_1} \cdot |R| - p_{i_1} \cdot (|R| - 1) + p_{i_2} \cdot (|R|-1) - p_{i_1} \cdot p_{i_2} \cdot (|R| -1)\\
    = \quad p&_{\min(R)} \cdot |R| + (|R|-1) \cdot (p_{i_2} - p_{i_1} - p_{i_1} \cdot p_{i_2}).
  \end{split}
\end{equation*}
It remains to show that $p_{i_2} - p_{i_1} - p_{i_1} \cdot p_{i_2} \geq 0$.
Notice that for a fixed $p_{i_2}$, the left-hand side is monotone decreasing with $p_{i_1}$.
Notice also that $p_i$ is monotone increasing with $i$. The smallest value of the expression $p_{i_2} - p_{i_1} - p_{i_1} \cdot p_{i_2}$ is therefore attained in the case $i_1 = i$, $i_2 = i+1$, for some $i \in [\dime-1]$.
It remains to verify:
\begin{equation*}
  \begin{split}
      p_{i_2} - p_{i_1} - p_{i_1} \cdot p_{i_2} &= \frac{1}{2(\dime-i)} - \frac{1}{2(\dime+1-i)} - \frac{1}{4(\dime-i)(\dime+1-i)}\\
    &=
      \frac{2\dime+2-2i-2\dime+2i -1}{4(\dime-i)(\dime+1-i)} = \frac{1}{4(\dime-i)(\dime+1-i)}
    \geq
      \frac{1}{4(\dime-1)\dime} > 0.
  \end{split}
\end{equation*}
This concludes the proof of the lemma.
\end{proof}

\subsection{Analysis of the Greedy Algorithm}
Observe that a feasible solution to the constructed instance is to open one facility in each vertex of the hypercube to serve the $\dime$ clients located in this vertex.
Let $R_v = \{(v,l): l \in [\dime]\}$ for every $v \in V$.
The cost of serving all clients from a vertex $v$ is
\begin{align*}
    g(R_v) &= \expected_{A \sim \pi}[|U(R_v,A)|]\\
  &=
    \prob_{A \sim \pi}[A = \emptyset] \cdot |U(R_v,\emptyset)| + \sum_{A \subseteq [\dime]: A \neq \emptyset} p_A \cdot |U(R_v,A)|\\
  &=
    \prob_{A \sim \pi}[A = \emptyset] \cdot 0 + \sum_{A \subseteq [\dime]: A \neq \emptyset} p_A \cdot 1 \leq 1.
\end{align*}
Hence, the cost of an optimum solution is at most $2^{\dime}$ in total, i.e., $\cost(\opt) \leq 2^{\dime}$.
We will now show that {\tt GreedySFL} produces a different and much more expensive solution whose structure is described in \cref{lem:greedy-edges}.

We call two clients $(v,l)$ and $(u,l)$ \emph{matching clients} if $\{v,u\} \in E$ and $v$ differs from $u$ on the coordinate $l$ (notice that both clients have the same index $l$).
Hence, there are $|E| = \dime \cdot 2^{\dime-1}$ pairs of matching clients.
For matching clients $c_1 = (v,l)$ and $c_2 = (u,l)$, we have $|U(\{c_1,c_2\},A)| \in \{0,1\}$ (because they share the same index), so $g(\{c_1,c_2\}) = p_l$.
Therefore, the cost-effectiveness of serving matching clients in a facility on edge $\{u,v\}$ is equal to $(g(\{c_1,c_2\}) + p_l)/2 = p_l$ and is also equal to the cost-effectiveness of serving $c_1$ in location $v$ as well as serving $c_2$ in location $u$.

\begin{lemma}\label{lem:greedy-edges}
 {\tt GreedySFL} opens all facilities on the edges of the hypercube and uses these facilities to serve pairs of matching clients from the vertices adjacent to the edges. 
\end{lemma}
\begin{proof}
  First, we show that {\tt GreedySFL} never selects new location in order to serve at least $3$ clients from at least $2$ different locations.
  Suppose, by contradiction, that {\tt GreedySFL} selected a subset of clients $R$ that contains at least $3$ clients from at least $2$ different locations $v$ and $u$ and decided to serve them from a new location, not previously used to serve any clients.
  Let $i$ be a coordinate in which $v$ differs from $u$.
  Consider splitting set $R = R_0^i \cup R_1^i$ according to dimension $i$ as in Lemma~\ref{l:dim_split}.
  By Lemma~\ref{l:dim_split}, the opening cost of serving $R_0^i$ and $R_1^i$ separately (i.e. $g(R_0^i)+g(R_1^i)$) is at most as much as the opening cost of serving them together (i.e. $g(R_0^i \cup R_1^i)$) increased by $\min\{|R_0^i|,|R_1^i|\} \cdot p_i$.
  Notice that, by splitting $R$ into $R_0^i$ and $R_1^i$ and serving them separately, we save at least $\min\{|R_0^i|,|R_1^i|\} \cdot p_i$ on the connection cost.
  Therefore it is not more expensive to serve $R_0^i$ and $R_1^i$ separately than to serve $R$.
  Therefore, at least one of the two sets $R_0^i$ or $R_1^i$ has a cost-effectiveness that is not greater than the cost-effectiveness of $R$.
  Furthermore, by applying Lemma~\ref{l:singleton}, we get that there exists a singleton set whose cost-effectiveness is not worse than $R$'s.
  Finally by the assumed tie-breaking rule, {\tt GreedySFL} favors singletons over sets of cardinality at least $3$ and hence {\tt GreedySFL} would choose the singleton set instead of $R$, which contradicts our assumption.

  Next, we show that {\tt GreedySFL} never selects clients which are served in already used locations.
  Suppose, by contradiction, that {\tt GreedySFL} decided a subset of clients $R$ to be served by a previously open facility $f$ already serving set of clients $T$, i.e. that clients from $R$ would join clients from $T$ already served at $f$ that is located on an edge that connect vertices $v$ and $u$ that differ in the coordinate $i$.
  Assume $T = \{(v,i), (u,i)\}$, i.e. $T$ is a pair of matching clients for the edge $\{v,u\}$.
  Note that we have $f = \{v,u\}$.
  If $R$ contains a client on a vertex $w \in V \setminus \{v, u\}$ as above we argue by Lemma~\ref{l:dim_split} that $R$ may be split into $R' \ni w$ and $R''$ containing clients from $\{v, u\}$ and that at least one of these sets would be preferred by {\tt GreedySFL} over $R$. Hence we may assume $R$ only contains clients located on either $v$ or $u$.
  Let us then denote by $R_v$ (resp. $R_u$) the sets of clients from $R$ that are located on $v$ (resp. $u$).
  Observe that $g(R_v \cup T) - g(T) = g(R_v \cup R_u \cup T) - g(R_u \cup T)$ and $g(R_u \cup T) - g(T) = g(R_u \cup R_v \cup T) - g(R_v \cup T)$, because $T$ is a pair of matching clients (who in the definition of $g(\cdot)$ are activated whenever the cube is collapsed in dimension $i$).
  Therefore, {\tt GreedySFL} might as well serve $R_v$ separately to $R_u$.
  We may therefore assume that $R$ only contains clients from a single vertex adjacent to the edge on which the facility is located.
  Note that the gain of $R$ from joining $T$ is $g(T) + g(R) - g(R \cup T) \leq g(T) = p_i$ (by subadditivity of $g(\cdot)$).
  In the case $|R| \geq 2$, this already shows that $R$ would not join $T$ because the facility cost gain of joining would be at least the connection cost of these clients traversing the distance of $p_i/2$.
  It remains to argue that a single client $c = (v,j)$ would also not be interested to join clients from $T$.
  To see this, observe that $g(T) = p_i$, $g(c) = p_j$ and $g(\{c_j\} \cup T) = p_i + p_j - p_i \cdot p_j$. Hence, the facility cost gain from joining $c$ to $T$ would be $p_i \cdot p_j \leq p_i / 2$ because $p_j \leq 1/2$ which is not more than the connection cost of $p_i / 2 = d(c,f)$ that client $c$ would have if it jointed $T$.

  We can show that {\tt GreedySFL} never selects a subset of more than one client from a single vertex of the hypercube to be served by a new facility.
  It is sufficient to use Lemma~\ref{l:singleton} to see that {\tt GreedySFL} favors a singleton set over a set of more clients from a single vertex.

  We have already shown that the solution produced by {\tt GreedySFL} is a collection of pairs of clients from different vertices and possibly some singletons.
  It remains to argue that {\tt GreedySFL} will not choose any singleton.
  Consider a client $c = (v,l)$ for some $v \in V, l \in [\dime]$.
  We will first show that it may only be served as a singleton or together with its matching client, i.e., client $c' = (u,l)$, where $u$ is the neighbor of $v$ in dimension $l$.
  Suppose $c$ was served together with $c'' \neq c'$.
  Observe that the cost-effectiveness of serving $\{c, c''\}$ is strictly worse than the minimal cost-effectiveness of serving $\{c\}$ and that of serving $\{c''\}$, hence {\tt GreedySFL} could not have chosen $\{c, c''\}$ as a set of clients.

  Now that we have excluded the possibility of not-matching clients to be served in pairs, we are left with matching pairs and singletons.
  It means that two matching clients are either served together or separately as singletons.
  It remains to recall that they are be equally cost-effective,
  so the tie-breaking rule chooses a matching pair before a singleton.
  Hence, {\tt GreedySFL} never selects a singleton.
\end{proof}
To conclude, we see that {\tt GreedySFL} produces a solution of cost
\begin{align*}
    \cost(\text{greedy})
    = &\sum_{i \in [\dime]} \:\:\sum_{\substack{\{v,u\} \in E:\\ v_i \neq u_i}} \Big( d(v,u) + g(\{(v,i),(u,i)\})\Big)\\
    = &\sum_{i \in [\dime]} \:\:\sum_{\substack{\{v,u\} \in E:\\ v_i \neq u_i}} \Big( p_i + \sum_{A \subseteq [\dime]} p_A \cdot |U(\{(v,i),(u,i)\},A)| \Big)\\
    = &\sum_{i \in [\dime]} \:\:\sum_{\substack{\{v,u\} \in E:\\ v_i \neq u_i}} \Big( p_i + \sum_{A \subseteq [\dime]: i \in A} p_A \cdot 1 \Big)\\
    = &\sum_{i \in [\dime]} \:\:\sum_{\substack{\{v,u\} \in E:\\ v_i \neq u_i}} \Big( p_i + p_i \Big)\\
    = &\sum_{i \in [\dime]} 2^{\dime-1} \cdot \frac{1}{\dime+1-i} = 2^{\dime-1} \cdot H_{\dime},
\end{align*}
where $H_i = 1+1/2+\cdots+1/i$ is the $i$th harmonic number.
The approximation ratio is then at least
\[
  \frac{\cost(\text{greedy})}{\cost(\opt)} \geq \frac{2^{\dime-1} \cdot H_{\dime}}{2^{\dime}} = \frac{H_{\dime}}{2} = \Omega(\log \dime).
\]
This holds for the instance with $n = \dime \cdot 2^{\dime}$ clients.
Therefore, $n < 4^{\dime}$ and $\log \log(n) < 2\log(\dime)$.
Hence, the ratio is then $\Omega(\log \log n)$ which finishes the proof of Theorem~\ref{thm:lower-bound}.

\end{document}